\documentclass[11pt]{article}

\usepackage[margin=1.2in]{geometry}
\usepackage{amsfonts,amssymb,amsmath,amsthm,mathtools,mathrsfs}
\usepackage{thmtools,thm-restate}

\usepackage{graphicx,microtype,enumitem,latexsym,lpic,bm,floatrow,xspace}
\usepackage[bf]{caption}

\usepackage{hyperref}

\usepackage[usenames,dvipsnames]{xcolor}
\hypersetup{
	colorlinks=true,
	linkcolor=Sepia,
	citecolor=Sepia,
	filecolor=Sepia,
	urlcolor=Sepia
}

\declaretheorem[name=Theorem]{theorem}
\declaretheorem[name=Lemma,sibling=theorem]{lemma}
\declaretheorem[name=Corollary,sibling=theorem]{corollary}
\declaretheorem[name=Definition,style=definition]{definition}

\declaretheorem[name=Example,style=remark]{example}

\newcommand{\N}{\mathbb{N}} % naturals
\newcommand{\Z}{\mathbb{Z}} % integers
\newcommand{\D}{\mathbb{D}}

\newcommand{\bs}{\mathrm{bs}}
\newcommand{\cbs}{\mathrm{cbs}}

\DeclareMathOperator*{\polylog}{polylog}
\DeclareMathOperator*{\argmax}{arg\,max}
\DeclareMathOperator{\poly}{poly}
\DeclareMathOperator*{\pct}{pct}
\newcommand{\Paths}{\text{Paths}}
\newcommand{\Viol}{\text{Viol}}
\newcommand{\vars}{\mathrm{vars}}
\newcommand{\cons}{\mathrm{cons}}
\newcommand{\Lines}{\text{Lines}}

\newcommand{\Peb}{\text{\slshape Peb}}
\newcommand{\Tse}{\text{\slshape Tse}}
\newcommand{\Sp}{\text{\itshape Sp}}

\newcommand{\newfunction}[2]{\newcommand{#1}{{\textrm{\upshape\small #2}}}}
\newfunction{\EQ}{EQ}
\newfunction{\ThreeEQ}{3EQ}
\newfunction{\AND}{AND}
\newfunction{\OR}{OR}
\newfunction{\XOR}{XOR}
\newfunction{\VER}{VER}
\newfunction{\HN}{HN}
\newfunction{\DISJ}{DISJ}
\newfunction{\UDISJ}{UDISJ}
\newfunction{\IND}{\text{-}IND}
\newfunction{\TwoIND}{2IND}
\newfunction{\ThreeIND}{3IND}
\newfunction{\Jump}{Jump}
\newfunction{\QCS}{QCS}
\newfunction{\GEN}{GEN}
\newfunction{\SAT}{SAT}

\renewcommand{\P}{\textsf{\upshape P}\xspace}
\newcommand{\NP}{\textsf{\upshape NP}\xspace}

\newcommand{\Th}[1]{\mathsf{Th}(#1)}
\newcommand{\Tcc}[1]{\mathsf{T}^{\mathsf{cc}}(#1)}
\newcommand{\LSfull}{Lov\'asz--Schrijver\xspace}
\newcommand{\LS}{\mathsf{LS}}

\newcommand{\CP}{\mathsf{CP}}

\begin{document}
% --------------------------------------------------------------------------- %

\newgeometry{left=1.3in,right=1.3in}
\vspace*{2ex}
\begin{center}
\LARGE

\renewcommand*{\thefootnote}{\fnsymbol{footnote}}

\textbf{Communication Lower Bounds via \\ Critical Block Sensitivity\footnote{A preliminary version of this work appeared in the Proceedings of STOC 2014~\cite{goos14communication}.}}

\renewcommand*{\thefootnote}{\arabic{footnote}}
\setcounter{footnote}{0}

\vspace*{4ex}

\large
Mika G\"o\"os and Toniann Pitassi

\medskip

Department of Computer Science \\ University of Toronto

\vspace*{4.5ex}
\end{center}
\vspace*{4ex}

\noindent
\textbf{Abstract.}\quad
We use \emph{critical block sensitivity}, a new complexity measure introduced by Huynh and Nordstr{\"o}m~({\small STOC~2012}), to study the communication complexity of search problems. To begin, we give a simple new proof of the following central result of Huynh and Nordstr{\"o}m: if $S$ is a search problem with critical block sensitivity~$b$, then every randomised two-party protocol solving a certain \emph{two-party lift} of $S$ requires $\Omega(b)$ bits of communication. Besides simplicity, our proof has the advantage of generalising to the multi-party setting. We combine these results with new critical block sensitivity lower bounds for \emph{Tseitin} and \emph{Pebbling} search problems to obtain the following applications.
\begin{itemize}
\item \textbf{Monotone circuit depth:}
We exhibit a monotone $n$-variable function in $\NP$ whose monotone circuits require depth $\Omega(n/\log n)$; previously, a bound of $\Omega(\sqrt{n})$ was known (Raz and Wigderson, {\small JACM~1992}). Moreover, we prove a $\Theta(\sqrt{n})$ monotone depth bound for a function in monotone $\P$.

\item \textbf{Proof complexity:}
We prove new rank lower bounds as well as obtain the first length--space lower bounds for semi-algebraic proof systems, including Lov\'asz--Schrijver and Lasserre (SOS) systems. In particular, these results extend and simplify the works of Beame et al.~({\small SICOMP~2007}) and Huynh and Nordstr{\"o}m.
\end{itemize}
\thispagestyle{empty}
\setcounter{page}{0}
\newpage
\restoregeometry

% =========================================================================== %
\section{Introduction}
% --------------------------------------------------------------------------- %

Apart from their intrinsic interest, communication lower bounds for \emph{search problems} find applications in two major areas of complexity theory.
\begin{itemize}
\item[\bf 1.] \textbf{Circuit complexity:}
A famous theorem of Karchmer and Wigderson~\cite{karchmer88monotone} states that for all boolean functions $f$, the minimum depth of a circuit computing $f$ is equal to the communication complexity of a certain search problem, called the \emph{Karchmer--Wigderson (KW) game} for~$f$. While it still remains a major open problem to prove general depth lower bounds for explicit boolean functions, KW-games have permitted progress in \emph{monotone} circuit complexity: there are monotone depth lower bounds for graph connectivity~\cite{karchmer88monotone}, clique functions~\cite{goldmann92simple,raz92monotone}, perfect matchings~\cite{raz92monotone}, and functions in monotone~\P~\cite{raz99separation}. See also Chapter 7 in Jukna's book~\cite{jukna12boolean}.

\item[\bf 2.] \textbf{Proof complexity:}
Impagliazzo et al.~\cite{impagliazzo94upper} (see also~\cite[\S19.3]{jukna12boolean}) introduced an analogue of KW-games to proof complexity. They showed how small tree-like Cutting Planes refutations of an unsatisfiable CNF formula $F$ can be converted into efficient \emph{two-party} communication protocols for a certain canonical search problem associated with~$F$. 
More recently, Beame et al.~\cite{beame07lower} extended this connection by showing that suitable lower bounds 
for \emph{multi-party} protocols imply degree/rank lower bounds for many well-studied semi-algebraic proof systems, including 
\LSfull~\cite{lovasz91cones}, Positivstellensatz~\cite{grigoriev01linear}, Sherali--Adams~\cite{sherali90hierarchy}, and Lasserre (SOS)~\cite{lasserre01explicit} systems. In parallel to these developments, Huynh and Nordstr{\"o}m~\cite{huynh12virtue} have also found a new kind of simulation of space-bounded proofs by communication protocols. They used this connection to prove length--space lower bounds in proof 
complexity.
\end{itemize}

In this work we obtain new randomised lower bounds for search problems in both two-party and multi-party settings. Our proofs are relatively simple reductions from the \emph{set-disjointness} function, the canonical $\NP$-complete problem in communication complexity. These results allow us to derive, almost for free, new lower bounds in the above two application domains.
\begin{itemize}
\item[\bf 1.] \textbf{Monotone depth:}
We introduce a certain monotone encoding of the \emph{CSP satisfiability} problem and prove an $\Omega(n/\log n)$ monotone depth lower bound for it, where $n$ is the number of input variables. Previously, the best bound for an explicit monotone function (perfect matchings) was $\Omega(\sqrt{n})$ due to Raz and Wigderson~\cite{raz92monotone}. Moreover, we prove a $\Theta(\sqrt{n})$ monotone depth bound for a function in monotone~$\P$.

\item[\bf 2.] \textbf{Rank, length, and space:}
We obtain new rank lower bounds for a family of semantic polynomial threshold proof systems called $\Tcc{k}$, which includes many of the semi-algebraic proof systems mentioned above. This extends and simplifies the work of Beame et al~\cite{beame07lower}. We also extend the length--space lower bound of Huynh and Nordstr{\"o}m~\cite{huynh12virtue} to hold for $\Tcc{k}$ systems of degree up to $k=(\log n)^{1-o(1)}$.
In particular, this yields the first nontrivial length--space lower bounds for dynamic SOS proofs of this degree.
\end{itemize}
We state these results more precisely shortly, once we first formalise our basic communication complexity setup.

\subsection{Starting point: Critical block sensitivity}

We build on the techniques recently introduced by Huynh and Nordstr{\"o}m \cite{huynh12virtue}. They defined a new complexity measure for search problems called \emph{critical block sensitivity}, which is a generalisation of the usual notion of block sensitivity for functions (see~\cite{buhrman02complexity} for a survey). They used this measure to give a general method of proving lower bounds for \emph{composed} search problems in the two-party communication model. These notions will be so central to our work that we proceed to define them immediately.

A \emph{search problem} on $n$ variables is a relation $S\subseteq\{0,1\}^n\times Q$ where $Q$ is some set of possible solutions. On input $\alpha\in\{0,1\}^n$ the search problem is to find a solution $q\in Q$ that is \emph{feasible for $\alpha$}, that is, $(\alpha,q)\in S$. We assume that $S$ is such that all inputs have at least one feasible solution. An input is called \emph{critical} if it has a unique feasible solution.

\begin{definition}[Critical block sensitivity~\cite{huynh12virtue}]
Fix a search problem $S\subseteq\{0,1\}^n\times Q$. Let $f\subseteq S$ denote a total function that solves $S$, i.e., for each input $\alpha\in\{0,1\}^n$ the function picks out some feasible solution $f(\alpha)$ for $\alpha$. We denote by $\bs(f,\alpha)$ the usual block sensitivity of $f$ at $\alpha$. That is, $\bs(f,\alpha)$ is the maximal number $\bs$ such that there are disjoint blocks of coordinates $B_1,\ldots,B_\bs\subseteq[n]$ satisfying $f(\alpha)\neq f(\alpha^{B_i})$ for all $i$; here, $\alpha^{B_i}$ is the same as $\alpha$ except the input bits in coordinates $B_i$ are flipped. The \emph{critical block sensitivity} of $S$ is defined as
\[
\cbs(S)\ :=\ \min_{f\subseteq S}\max_{\text{critical } \alpha} \bs(f,\alpha).
\]
\end{definition}

We note immediately that $\cbs(S)$ is a lower bound on the deterministic decision tree complexity of $S$. Indeed, a deterministic decision tree defines a total function $f\subseteq S$ and on each critical input $\alpha$ the tree must query at least one variable from each sensitive block of $f$ at $\alpha$ (see~\cite[Theorem 9]{buhrman02complexity}). It turns out that $\cbs(S)$ is also a lower bound on the \emph{randomised} decision tree complexity (see \autoref{thm:two-party} below).

\subsection{Composed search problems}

In order to study a search problem $S\subseteq\{0,1\}^n\times Q$ in the setting of two-party communication complexity, we need to specify how the $n$ input variables of $S$ are divided between the two players, Alice and Bob. 

Unfortunately, for many search problems (and functions) there is often no partition of the variables that would carry the ``intrinsic'' complexity of $S$ over to communication complexity. For example, consider computing the $\AND$ function on $n$ inputs. The block sensitivity of $\AND$ is $n$, but this complexity is lost once we move to the two-party setting: only $O(1)$ many bits need to be communicated between Alice and Bob regardless of the input partition.

For this reason, one usually studies \emph{composed} (or \emph{lifted}) variants $S\circ g^n$ of the original problem; see \autoref{fig:lift}. In a composed problem, each of the $n$ input bits of $S$ are encoded using a small two-party function $g\colon\mathcal{X}\times\mathcal{Y}\to\{0,1\}$, sometimes called a \emph{gadget}. As input to $S\circ g^n$  Alice gets an $x\in\mathcal{X}^n$ and Bob gets a $y\in\mathcal{Y}^n$. We think of the pair $(x,y)$ as encoding the input
\[
\alpha = g^n(x,y) = (\,g(x_1,y_1),\ldots,g(x_n,y_n)\,)
\]
of the original problem $S$. The objective is to find a $q\in Q$ such that $(g^n(x,y),q)\in S$.

\begin{figure}[t]
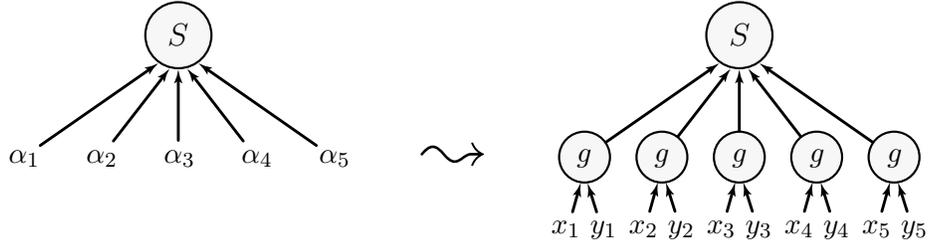
%
\caption{Composing a search problem $S$ with a two-party gadget $g$.}
\label{fig:lift}
\begin{lpic}[b(-2mm)]{lift(.27)}
\lbl[c]{88.8,111;\large$S$}
\lbl[c]{13,50;$\alpha_1$}
\lbl[c]{51.25,50;$\alpha_2$}
\lbl[c]{89.5,50;$\alpha_3$}
\lbl[c]{127.75,50;$\alpha_4$}
\lbl[c]{166,50;$\alpha_5$}
\lbl[c]{225,50;\Huge$\leadsto$}

\lbl[c]{365.25,111;\large$S$}
\lbl[c]{288.75,50;$g$}
\lbl[c]{327,50;$g$}
\lbl[c]{365.25,50;$g$}
\lbl[c]{403.5,50;$g$}
\lbl[c]{441.75,50;$g$}
\lbl[c]{288.75,15;$x_1\ y_1$}
\lbl[c]{327,15;$x_2\ y_2$}
\lbl[c]{365.25,15;$x_3\ y_3$}
\lbl[c]{403.5,15;$x_4\ y_4$}
\lbl[c]{441.75,15;$x_5\ y_5$}
\end{lpic}
\end{figure}

\subsection{Our communication complexity results} \label{ssec:results}

We start by giving a simple new proof of the following central result of Huynh and Nordstr{\"o}m~\cite{huynh12virtue}. (Strictly speaking, the statement of the original theorem~\cite{huynh12virtue} is slightly weaker in that it involves an additional ``consistency'' assumption, which we do not need.)
\begin{restatable}[Two-party version]{theorem}{twoparty} \label{thm:two-party}
There is a two-party gadget $g\colon \mathcal{X}\times\mathcal{Y}\to\{0,1\}$ such that if $S\subseteq\{0,1\}^n\times Q$ is any search problem, then $S\circ g^n$ has randomised bounded-error communication complexity $\Omega(\cbs(S))$.
\end{restatable}
Huynh and Nordstr{\"o}m proved \autoref{thm:two-party} for the gadget $g=\ThreeIND$, where $\ThreeIND\colon[3]\times\{0,1\}^3\to\{0,1\}$ is the indexing function that maps $(x,y)\mapsto y_x$. Their proof used the information complexity approach~\cite{chakrabarti01informational,bar-yossef04information} and is quite intricate. By contrast, we prove \autoref{thm:two-party} by a direct randomised reduction from the \emph{set-disjointness} function
\[\textstyle
\DISJ_n(x,y) = (\OR_n\circ \AND^n)(x,y) = \bigvee_{i\in[n]}(x_i\land y_i).
\]
In the language of Babai et al.~\cite{babai86complexity} (see also~\cite{chattopadhyay10story}) the set-disjointness function is \NP-complete in communication complexity: it is easy to certify that $\DISJ_n(x,y)=1$, and conversely, every two-party function with low nondeterministic complexity reduces efficiently to $\DISJ_n$. Our proof of \autoref{thm:two-party} is inspired by a result of Zhang~\cite{zhang09tightness} that essentially establishes \autoref{thm:two-party} in case $S$ is a function and $\cbs(S)$ is simply the standard block sensitivity. The key insight in our proof is to choose $g$ to be \emph{random-self-reducible} (see \autoref{sec:versatile} for definitions). Random-self-reducibility is a notion often studied in cryptography and classical complexity theory, but less often in communication complexity. Most notably, random-self-reducibility was used implicitly in~\cite{raz92monotone}. The definitions we adopt are similar to those introduced by Feige et al.~\cite{feige94minimal} in a cryptographic context.

Our proof has also the advantage of generalising naturally to the multi-party setting. This time we start with the $k$-party unique-disjointness function $\UDISJ_{k,n}$ and the proof involves the construction of $k$-party random-self-reducible functions $g_k$.
\begin{theorem}[Multi-party version] \label{thm:multi-party}
There are $k$-party gadgets $g_k\colon \mathcal{X}^k\to\{0,1\}$ with domain size $\log|\mathcal{X}|=k^{o(1)}$ bits per player, such that if $S\subseteq\{0,1\}^n\times Q$ is any search problem, then $S\circ g_k^n$ has randomised bounded-error communication complexity at least that of $\UDISJ_{k,\cbs(S)}$ (up to constants).
\end{theorem}

\autoref{thm:multi-party} can be applied to the following multi-player communication models.
\begin{itemize}[label=$-$]
\item {\bf Number-in-hand:}
The $i$-th player only sees the $i$-th part of the input. Here, set-disjointness has been studied under broadcast communication (e.g.,~\cite{gronemeier09asymptotically}) and under private channel communication~\cite{braverman13tight}.

\item {\bf Number-on-forehead (NOF):}
The $i$-th player sees all parts of the input except the $i$-th part~\cite{chandra83multi}. The current best randomised lower bound for $\UDISJ_{k,n}$ is $\Omega(\sqrt{n}/2^kk)$ by Sherstov~\cite{sherstov13communication}. We rely heavily on Sherstov's result in our proof complexity applications.
\end{itemize}

In the rest of this introduction we discuss the applications---the impatient reader who wants to see the proofs of Theorems~\ref{thm:two-party} and \ref{thm:multi-party} can immediately skip to Sections~\ref{sec:versatile} and \ref{sec:cc-lb}.

\subsection{CSPs and their canonical search problems}

To get the most out of Theorems~\ref{thm:two-party} and \ref{thm:multi-party} for the purposes of applications, we need to find search problems with high critical block sensitivity but low certificate complexity. Low-degree constraint satisfaction problems (CSPs) capture exactly the latter goal~\cite{lovasz95search}.

\begin{definition}[$d$-CSPs]
A CSP $F$ consists of a set of (boolean) variables $\vars(F)$ and a set of constraints $\cons(F)$. Each constraint $C\in\cons(F)$ is a function that maps a truth assignment $\alpha\colon\vars(F)\to\{0,1\}$ to either $0$ or $1$. If $C(\alpha)=1$, we say that $C$ is \emph{satisfied} by $\alpha$, otherwise $C$ is \emph{violated} by $\alpha$. Let $\vars(C)$ denote the smallest subset of $\vars(F)$ such that $C$ depends only on the truth values of the variables in $\vars(C)$. We say that $F$ is of \emph{degree}~$d$, or $F$ is a $d$-CSP, if $|\vars(C)|\leq d$ for all $C$. Note that $d$-CNF formulas are a special case of $d$-CSPs, and conversely, each $d$-CSP can be written as an equivalent $d$-CNF with a factor $2^d$ blow-up in the number of constraints.
\end{definition}

An \emph{unsatisfiable} CSP $F$ has no assignment that satisfies all the constraints. Each such $F$ comes with an associated \emph{canonical search problem} $S(F)$.
\begin{definition}[Canonical search problems]
Let $F$ be an unsatisfiable CSP. In the search problem $S(F)$ we are given an assignment $\alpha\colon\vars(F)\to\{0,1\}$ and the goal is to find a constraint $C\in\cons(F)$ that is violated by $\alpha$.
\end{definition}

We give new critical block sensitivity lower bounds for the canonical search problems associated with \emph{Tseitin} and \emph{Pebbling} formulas.

\subsection{Sensitivity of Tseitin formulas}

Tseitin formulas are well-studied examples of unsatisfiable CSPs that are hard to refute in many proof systems; for an overview, see Jukna~\cite[\S18.7]{jukna12boolean}.

\begin{definition}[Tseitin formulas]
Let $G=(V,E,\ell)$ be a connected labelled graph of maximum degree $d$ where the labelling $\ell\colon V\to\{0,1\}$ has odd Hamming weight. The \emph{Tseitin formula} $\Tse_G$ associated with $G$ is the $d$-CSP that has the edges $e\in E$ as variables and for each node $v\in V$ there is a constraint $C_v$ defined by
\[
C_v(\alpha) = 1
\quad\iff\quad
\sum_{e: v\in e} \alpha(e) \equiv \ell(v) \pmod{2}.
\]
It follows from a simple parity argument that $\Tse_G$ is unsatisfiable (see, e.g., \autoref{sec:tseitin}).
\end{definition}

Call $G$ \emph{$\kappa$-routable} if there is a set $T\subseteq V$ of size $|T|\geq 2\kappa$ such that for any set of $\kappa$ disjoint pairs of nodes of $T$ there are $\kappa$ edge-disjoint paths in $G$ that connect all the pairs. (Note: $\kappa$-routability is usually defined only for $T=V$, but we relax this condition.) The proof of the following theorem appears in \autoref{sec:cbs}.
\begin{restatable}[Tseitin sensitivity]{theorem}{tseitin} \label{thm:tseitin}
If $G$ is $\kappa$-routable, then $\cbs(S(\Tse_G))=\Omega(\kappa)$.
\end{restatable}

\autoref{thm:tseitin} can be applied to the following classes of bounded-degree graphs.
\begin{itemize}[label=$-$]
\item {\bf Grid graphs:}
If $G$ is a $\sqrt{n}\times \sqrt{n}$ grid graph, then we can take $\kappa = \Omega(\sqrt{n})$ by letting $T\subseteq V$ be any row (or column) of nodes. This is tight: the deterministic decision tree that solves $S(\Tse_G)$ using binary search makes $O(\sqrt{n})$ queries.
	
\item {\bf Expanders:}
If $G$ is a sufficiently strong expander (e.g., a Ramanujan graph~\cite{lubotzky88ramanujan}), then we can take $\kappa = \Omega(n/\log n)$ as shown by Frieze et al.~\cite{frieze00optimal,frieze01edge}.

\item {\bf Connectors:}
A \emph{$\kappa$-connector} is a bounded-degree graph with $\kappa$ inputs $I\subseteq V$ and $\kappa$ outputs $O\subseteq V$ such that for any one-to-one correspondence $\pi\colon I\to O$ there exist $\kappa$ edge-disjoint paths that connect $i\in I$ to $\pi(i)\in O$. If we merge $I$ and $O$ in a $2\kappa$-connector in some one-to-one manner and let $T=I=O$, we get a $\kappa$-routable graph. Conversely, if $G$ is $\kappa$-routable, we can partition the set $T$ as $I\cup O$ and get a $\kappa$-connector.

It is known that simple $\kappa$-connectors with $\kappa=\Theta(n/\log n)$ exist and this bound is the best possible~\cite{pippenger90communication}. Thus, the best lower bound provable using \autoref{thm:tseitin} is $\Theta(n/\log n)$.
\end{itemize}

It is well known that the \emph{deterministic} decision tree complexity of $S(\Tse_G)$ is $\Omega(n)$ when $G$ is an expander~\cite{urquhart87hard}. However, \emph{randomised} lower bounds---which \autoref{thm:tseitin} provides---are more scarce. We are only aware of a single previous result in the direction of \autoref{thm:tseitin}, namely, Lov{\'a}sz et al.~\cite[\S3.2.1]{lovasz95search} announce a lower bound of $\Omega(n^{1/3})$ for the randomised decision tree complexity of $S(\Tse_G)$ when $G$ is an expander. Our \autoref{thm:tseitin} subsumes this.

\subsection{Sensitivity of pebbling formulas}

Pebble games have been studied extensively as means to understand time and space in computations; for an overview, see the survey by Nordstr{\"o}m~\cite{nordstrom13pebble}. In this work we restrict our attention to the simple \emph{(black) pebble game} that is played on a directed acyclic graph~$G$ with a unique sink node $t$ (i.e., having outdegree $0$). In this game the goal is to place a pebble on the sink $t$ using a sequence of \emph{pebbling moves}. The allowed moves are:
\begin{enumerate}[label=(\arabic*),noitemsep]
\item A pebble can be placed on a node if its in-neighbours have pebbles on them. In particular, we can always pebble a source node (i.e., having indegree $0$).
\item A pebble can be removed from any pebbled node (and reused later in the game).
\end{enumerate}
The \emph{(black) pebbling number} of $G$ is the minimum number of pebbles that are needed to pebble the sink node in the pebble game on $G$.

The pebble game on $G$ comes with an associated \emph{pebbling formula}.
\begin{definition}[{Pebbling formulas. See~\cite{ben-sasson01short} and~\cite[\S2.3]{nordstrom13pebble}}] \label{def:pebbling}
Let $G=(V,E,t)$ be a directed acyclic graph of maximum indegree $d$ where $t$ is a unique sink. The \emph{pebbling formula} $\Peb_G$ associated with $G$ is the $(d+1)$-CSP that has the nodes $v\in V$ as variables and the following constraints:
\begin{enumerate}[label=(\arabic*),noitemsep]
\item The variable corresponding to the sink $t$ is false.
\item For all nodes $v$ with in-neighbours $w_1,\ldots,w_d$, we require that if all of $w_1,\ldots,w_d$ are true, then $v$ is true. In particular, each source node must be true.
\end{enumerate}
It is not hard to see that $\Peb_G$ is unsatisfiable.
\end{definition}

Classical complexity measures for $S(\Peb_G)$ include the pebbling number of $G$ (a measure of \emph{space}) and the deterministic decision tree complexity (a measure of \emph{parallel time}), which admits many equivalent characterisations~\cite{chan13just}. However, these complexity measures are fundamentally \emph{deterministic} and do not seem to immediately translate into \emph{randomised} lower bounds, which are needed in our applications. For this reason, Huyhn and Nordstr{\"o}m~\cite{huynh12virtue} devised an elegant ad hoc proof method for their result that, for a \emph{pyramid graph} $G$ (see~\autoref{fig:pyramid}), $\cbs(S(\Peb_G))=\Omega(n^{1/4})$. Annoyingly, this falls a little short of both the pebbling number $\Theta(\sqrt{n})$ of $G$ and the decision tree complexity $\Theta(\sqrt{n})$ of $S(\Peb_G)$. Here we close this gap by generalising their proof method: we get tight bounds for a different (but related) graph~$G$. The proof appears in \autoref{sec:cbs}.
\begin{restatable}[Pebbling sensitivity]{theorem}{pebbling} \label{thm:pebbling}
There are bounded-degree graphs $G$ on $n$ nodes such that
\begin{itemize}[label=$-$,noitemsep]
\item $G$ has pebbling number $\Theta(\sqrt{n})$.
\item $S(\Peb_G)$ has deterministic decision tree complexity $\Theta(\sqrt{n})$.
\item $S(\Peb_G)$ has critical block sensitivity $\Theta(\sqrt{n})$.
\end{itemize}
\end{restatable}

\subsection{Applications: Monotone depth} \label{sec:app-monotone}

\paragraph{Monotone depth from Tseitin.}
Let $G$ be an $\Omega(n/\log n)$-routable graph of bounded degree \mbox{$d=O(1)$}. By \autoref{thm:tseitin} the lifted problem $S(\Tse_G)\circ g^{O(n)}$ has two-party communication complexity~$\Omega(n/\log n)$. By contrast, its nondeterministic communication complexity is just $\log n + O(1)$, since the players can guess a node $v\in V(G)$ and verify that it indeed induces a parity violation (which involves exchanging the inputs to $d=O(1)$ many copies of $g$ associated to edges incident to $v$). It is known that \emph{any} two-party search problem with nondeterministic communication complexity $C$ reduces to a monotone KW-game for some monotone~$f\colon \{0,1\}^N\to\{0,1\}$ on $N=2^C$ variables; see G{\'a}l~\cite[Lemma~2.3]{gal01characterization} for an exposition. In our case we get a monotone function on $N=O(n)$ variables whose monotone KW-game complexity---i.e., its monotone depth complexity---is $\Omega(N/\log N)$. Moreover, we make this general connection a bit more explicit in \autoref{sec:monotone-proofs} by showing that our function can be taken to be a monotone variant of the usual \emph{CSP satisfiability} function.
\begin{corollary}[Monotone depth from Tseitin] \label{cor:monotone-dept-tseitin}
There is an monotone function in $\NP$ on~$N$ input bits whose monotone depth complexity is $\Omega(N/\log N)$.
\end{corollary}

\paragraph{Monotone depth from pebbling.}
We also get perhaps the simplest proof yet of a $n^{\Omega(1)}$ monotone depth bound for a function in monotone $\P$. Indeed, we only need to apply a transformation of Raz and McKenzie, described in~\cite[\S3]{raz99separation}, which translates our $\Omega(\sqrt{n})$ communication lower bound for $S(\Peb_G)\circ g^{O(n)}$ (coming from Theorems \ref{thm:two-party} and \ref{thm:pebbling}) to a monotone depth lower bounds for a related ``generation'' function $\GEN_{G'}$ defined relative to a ``lifted'' version $G'$ of~$G'$. Raz and McKenzie originally studied the case when $G$ is a pyramid graph, and they lifted $S(\Peb_G)$ with some $\poly(n)$-size gadget (making the number of input bits of $\GEN_{G'}$ a large polynomial in $n$). However, their techniques work for any graph $G$ and any gadget $g$. In our case of constant-size gadgets, we get only constant factor blow-up in parameters; we refer to~\cite[\S3]{raz99separation} for the details of deriving the following.%
\begin{corollary}[Monotone depth from pebbling] \label{cor:pebbling}
There is an explicit function $f$ on $N$ input bits such that $f$ admits polynomial size monotone circuits of depth $\Theta(\sqrt{N})$ and any monotone circuit for $f$ requires depth $\Theta(\sqrt{N})$. \qed
\end{corollary}
The original bounds of \cite{raz99separation} went up to $\Omega(N^{\delta})$ for a small constant $\delta$. This was recently improved by the works~\cite{chan12tight,filmus13average} that prove (among other things) monotone depth bounds of up to $\Omega(N^{1/6-o(1)})$ for $\GEN_G$ type functions. Our \autoref{cor:pebbling} achieves quantitatively the largest bound (to-date) for a function in monotone $\P$.

\subsection{Applications: Proof complexity} \label{sec:app-proof}

Over the last decade or so there have been a large number of results proving lower bounds on the rank required to refute (or approximately optimise over) systems of constraints in a wide variety of semi-algebraic (a.k.a.\ polynomial threshold) proof systems, including \LSfull~\cite{lovasz91cones}, Cutting Planes~\cite{gomory58outline,chvatal73edmonds}, Positivstellensatz~\cite{grigoriev01linear}, Sherali--Adams~\cite{sherali90hierarchy}, and Lasserre~\cite{lasserre01explicit} proofs. Highlights of this work include recent linear rank lower bounds for many constraint optimisation problems~\cite{schoenebeck08linear,tulsiani09csp,charikar09integrality,schoenebeck07linear,georgiou10integrality}. Nearly all of these results rely on delicate constructions of local distributions that are specific to both the problem and to the proof system. 

A communication complexity approach for proving lower bounds for semi-algebraic proofs was developed by Beame et al.~\cite{beame07lower}.
They studied a semantic proof system called $\Tcc{k}$ whose proofs consist of lines that are computed by a low-cost (i.e., $\polylog$ communication) $k$-party NOF protocols (see \autoref{sec:proof-complexity} for definitions). They prove that if a CNF formula $F$ has a small tree-like $\Tcc{k}$ refutation, then $S(F)$ has an efficient $k$-party
NOF protocol. Thus, lower bounds for the tree-size of $\Tcc{k}$ proofs
follow from NOF lower bounds for $S(F)$.

\paragraph{Rank lower bounds.}
Using this relationship we can now prove the following result%
\footnote{Similar claims were made in \cite{beame10hardness}. Unfortunately, as pointed out by~\cite{huynh12virtue}, Lemma 3.5 in~\cite{beame10hardness} is incorrect and this renders many of the theorems in the paper incorrect.}
for $\Tcc{k}$ proof systems, where $k$ can be almost logarithmic in the size of the formula. We state the theorem only for rank, with the understanding that a bound of $\Omega(R)$ on rank also implies a bound of $\exp(\Omega(R))$ on tree-size. The proof appears in \autoref{sec:proof-complexity}.
\begin{restatable}[Rank lower bounds]{theorem}{rankbounds} \label{thm:rank-lb}
There are explicit {CNF} formulas $F$ of size $s$ and width $O(\log s)$ such that all $\Tcc{k}$ refutations of $F$ require rank at least
\[
R_k(s) =
\begin{cases}
s^{1-o(1)}, &\text{for}\enspace k=2, \\
s^{1/2-o(1)}, &\text{for}\enspace 3\leq k\leq(\log s)^{1-o(1)}.
\end{cases}
\]
\end{restatable}
\autoref{thm:rank-lb} simplifies the proof of a similar theorem from \cite{beame07lower}, which held only for a specific family of formulas obtained from non-constant degree graphs, and only for $k < \log \log s$.

We note already here that the quadratic gap between $R_2(s)$ and $R_3(s)$ will be an artefact of us switching from two-party communication to multi-party communication. More specifically, while the two-party communication complexity of set-disjointness $\DISJ_n$ is $\Omega(n)$, the corresponding lower bound for three parties is only~$\Omega(\sqrt{n})$~\cite{sherstov13communication}. Whether the multi-party bound can be improved to $\Omega(n)$ is an open problem.

\paragraph{Length--space lower bounds.}
Continuing in similar spirit, \cite{huynh12virtue} showed how to prove length--space lower bounds for $\Tcc{2}$ systems from lower bounds on the communication complexity of $S(F)$. Using this relationship together with our new multi-party lower bounds, we can extend this result to $\Tcc{k}$ systems of degree $k>2$.
\begin{restatable}[Length--space lower bounds]{theorem}{tradeoffs} \label{thm:tradeoffs}
There are CNF formulas $F$ of size $s$ such that
\begin{itemize}[label=$-$,noitemsep]
\item $F$ admits a Resolution refutation of length $L=s^{1+o(1)}$ and space $\Sp= s^{1/2+o(1)}$.
\item Any length $L$ and space $\Sp$ refutation of $F$ in $\Tcc{k}$ must satisfy
\begin{equation} \label{eq:tradeoff}
\Sp \cdot \log L \enspace\geq\enspace
\begin{cases}
s^{1/2-o(1)}, &\text{for}\enspace k=2, \\
s^{1/4-o(1)}, &\text{for}\enspace 3\leq k\leq(\log s)^{1-o(1)}.
\end{cases}
\end{equation}
\end{itemize}
\end{restatable}
We hesitate to call \autoref{thm:tradeoffs} a \emph{tradeoff} result since our only upper bound is a refutation requiring space $\Sp=s^{1/2+o(1)}$ and we do not know how to decrease this space usage by trading it for length; this is the same situation as in~\cite{huynh12virtue}. Surprisingly, in a subsequent work, Galesi et al.~\cite{galesi15space} have shown that \emph{any} unsatisfiable CNF formula admits an exponentially long Cutting Planes refutation in \emph{constant} space, which gives a second data point in the length--space parameter space for which an upper bound exists. We also mention that while the CNF formulas $F$ in \autoref{thm:tradeoffs} are lifted versions of pebbling formulas, we could have formulated similar length--space lower bounds for lifted Tseitin formulas (where, e.g., $\Sp\cdot\log L \geq s^{1-o(1)}$ for $k=2$). But for Tseitin formulas we do not have close-to-matching upper bounds.

In any case, \autoref{thm:tradeoffs} gives, in particular, the first length--space lower bounds for dynamic SOS proofs of degree $k$. In addition, even in the special case of $k=2$, \autoref{thm:tradeoffs} simplifies and improves on \cite{huynh12virtue}. However, for Polynomial Calculus Resolution (a $\Tcc{2}$ system), the best known length--space tradeoff results are currently proved in the recent work of Beck et al.~\cite{beck13some}. For Resolution (maybe the simplest $\Tcc{2}$ system), even stronger tradeoff results have been known since~\cite{ben-sasson11understanding}; see also Beame et al.~\cite{beame12time} for nontrivial length lower bounds in the superlinear space regime. For Cutting Planes (a $\Tcc{2}$ system) \autoref{thm:tradeoffs} remains the state-of-the-art to our knowledge.

\subsection{Models of communication complexity}

We work in the standard models of two-party and multi-party communication complexity; see~\cite{kushilevitz97communication,jukna12boolean} for definitions. Here we only recall some conventions about randomised protocols. A protocol $\Pi$ solves a search problem $S$ with \emph{error $\epsilon$} iff on any input $x$ the probability that $(x,\Pi(x))\in S$ is at least $1-\epsilon$ over the random coins of the protocol. Note that $\Pi(x)$ need not be the same feasible solution; it can depend on the outcomes of the random coins. The protocol is of \emph{bounded-error} if $\epsilon \leq 1/4$. The constant $1/4$ here can often be replaced with any other constant less than $1/2$ without affecting the definitions too much. In the case of computing boolean functions this follows from standard boosting techniques~\cite[Exercise 3.4]{kushilevitz97communication}. While these boosting techniques may fail for general search problems, we do not encounter any such problems in this work.

% =========================================================================== %
\section{Versatile Gadgets} \label{sec:versatile}
% --------------------------------------------------------------------------- %

In this section we introduce \emph{versatile} two-party and multi-party functions. Our proofs of Theorems~\ref{thm:two-party} and \ref{thm:multi-party} will work whenever we choose $g$ or $g_k$ to be a versatile gadget. We start by introducing the terminology in the two-party case; the multi-party case will be analogous.

\subsection{Self-reductions and versatility}
The simplest reductions between communication problems are those that can be computed without communication. Let $f_i\colon \mathcal{X}_i\times\mathcal{Y}_i\to\{0,1\}$ for $i=1,2$, be two-party functions. We say that $f_1$ \emph{reduces to} $f_2$, written $f_1\leq f_2$, if the communication matrix of $f_1$ appears as a submatrix of the communication matrix of $f_2$. Equivalently, $f_1\leq f_2$ iff there exist one-to-one mappings $\pi_A$ and $\pi_B$ such that
\[
f_1(x,y) = f_2(\pi_A(x),\pi_B(y))\qquad\text{for all}\quad (x,y)\in\mathcal{X}_1\times\mathcal{Y}_1.
\]
Our restriction to one-to-one reductions above is merely a technical convenience (cf.\ Babai et al.~\cite{babai86complexity} allow reductions to be many-to-one).

\begin{example}
Let $\ThreeEQ\colon[3]\times[3]\to\{0,1\}$ be the equality function with inputs from $[3]$. Then $\AND$ reduces to $\ThreeEQ$ since $\AND(x,y)=\ThreeEQ(1+x,3-y)$.
\end{example}
We will be interested in special kinds of reductions that reduce a function to \emph{itself}. Our first flavour of self-reducibility relates a function $f$ and its negation $\neg f$:
\begin{itemize}[leftmargin=1.4cm]
	\item[{\raisebox{-.63\height}[0pt][0pt]{
	\includegraphics[scale=0.13]{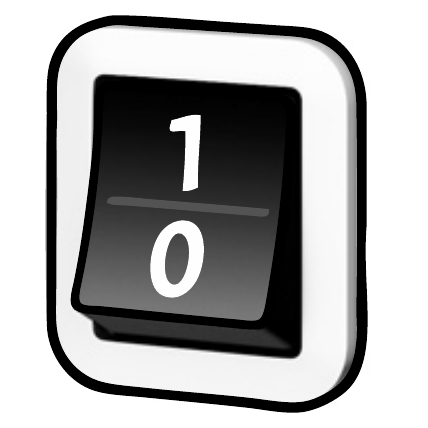}
	}}\hspace{-3pt}]
{\bf Flippability.} A function $f$ is called \emph{flippable} if $\neg f\leq f$. Note that since the associated reduction maps $z$-inputs to $(1-z)$-inputs in a one-to-one fashion, a flippable function must be \emph{balanced}: exactly half of the inputs satisfy $f(x,y)=1$.
\end{itemize}

\begin{example}
The $\XOR$ function is flippable via $\neg\XOR(x,y)=\XOR(1-x,y)$. By contrast, $\AND$ and $\ThreeEQ$ are not balanced and hence not flippable.
\end{example}

We will also consider randomised reductions where the two parties are allowed to \emph{synchronise} their computations using public randomness. More precisely, even though the two parties are still not communicating, we can let the mappings $\pi_A$ and $\pi_B$ depend on a public random string $\bm{r}\in\{0,1\}^*$, whose distribution the two parties can freely choose. This way, a random reduction computes $(x,y) \mapsto (\pi_A(x,\bm{r}),\pi_B(y,\bm{r}))$. The following definition is similar to the \emph{perfectly secure} functions of Feige et al.~\cite{feige94minimal}.
\begin{itemize}[leftmargin=1.4cm]
	\item[{\raisebox{-.6\height}[0pt][0pt]{
	\includegraphics[scale=0.12]{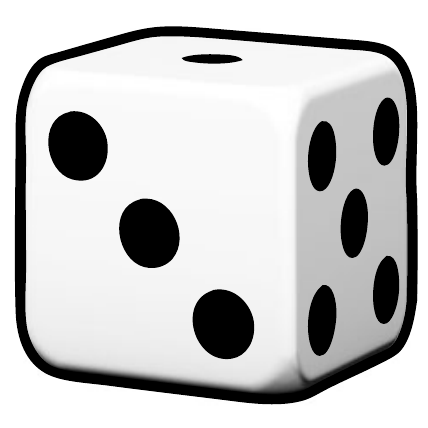}
	}}\hspace{-1pt}]
{\bf Random self-reducibility.} A function $f$ is called \emph{random-self-reducible} if there are mappings $\pi_A$ and $\pi_B$ together with a random variable $\bm{r}$ such that for every $z$-input $(x,y)\in f^{-1}(z)$ the random pair $(\pi_A(x,\bm{r}),\pi_B(y,\bm{r}))$ is uniformly distributed among all the $z$-inputs of $f$.
\end{itemize}
\begin{example}
The equality function $\EQ\colon [n]\times[n]\to\{0,1\}$ is random-self-reducible: we can use the public randomness to sample a permutation $\bm{\pi}\colon[n]\to[n]$ uniformly at random and let the two parties compute $(x,y)\mapsto(\bm{\pi}(x),\bm{\pi}(y))$. (In fact, to further save on the number of random bits used, it would suffice to choose $\bm{\pi}$ from any group that acts 2-transitively on $[n]$.)
\end{example}

A notable example of a function that is \emph{not} random-self-reducible is $\AND$; it has only one $1$-input, which forces any self-reduction to be the identity map. This is particularly inconvenient since $\AND$ is featured in the set-disjointness function $\DISJ_n=\OR_n\circ\AND^n$, which will be the starting point for our reductions. To compensate for the shortcomings of $\AND$ we work with a slightly larger function $g\geq\AND$ instead.
\begin{definition}[Versatility]
A two-party function $g$ is called \emph{versatile} if (1) $g\geq\AND$, (2) $g$ is flippable, and (3) $g$ is random-self-reducible.
\end{definition}

\begin{figure}
\begin{floatrow}
\ffigbox[6cm]{
\begin{lpic}[]{versatile(.25)}
\lbl[c]{15.8 ,15.8;$1$}
\lbl[c]{41.2 ,15.8;$0$}
\lbl[c]{66.66,15.8;$0$}
\lbl[c]{92   ,15.8;$1$}

\lbl[c]{15.8 ,41.2;$1$}
\lbl[c]{41.2 ,41.2;$1$}
\lbl[c]{66.66,41.2;$0$}
\lbl[c]{92   ,41.2;$0$}

\lbl[c]{15.8 ,66.66;$0$}
\lbl[c]{41.2 ,66.66;$1$}
\lbl[c]{66.66,66.66;$1$}
\lbl[c]{92   ,66.66;$0$}

\lbl[c]{15.8 ,92;$0$}
\lbl[c]{41.2 ,92;$0$}
\lbl[c]{66.66,92;$1$}
\lbl[c]{92   ,92;$1$}

\boldmath
\lbl[c]{-8 ,92;$0$}
\lbl[c]{-8 ,66.66;$1$}
\lbl[c]{-8 ,41.2;$2$}
\lbl[c]{-8 ,15.8;$3$}

\lbl[c]{15.8 ,116;$0$}
\lbl[c]{41.2 ,116;$1$}
\lbl[c]{66.66,116;$2$}
\lbl[c]{92   ,116;$3$}
\end{lpic}}
{\caption{Function $\VER$.}\label{fig:ver}}

\ffigbox[6cm]{
\begin{lpic}[]{hnmatrix(.25)}
\lbl[c]{15.8 ,15.8;$0$}
\lbl[c]{41.2 ,15.8;$0$}
\lbl[c]{66.66,15.8;$1$}
\lbl[c]{92   ,15.8;$1$}
\lbl[c]{117.4,15.8;$1$}
\lbl[c]{142.8,15.8;$0$}

\lbl[c]{15.8 ,41.2;$0$}
\lbl[c]{41.2 ,41.2;$1$}
\lbl[c]{66.66,41.2;$0$}
\lbl[c]{92   ,41.2;$1$}
\lbl[c]{117.4,41.2;$0$}
\lbl[c]{142.8,41.2;$1$}

\lbl[c]{15.8 ,66.66;$1$}
\lbl[c]{41.2 ,66.66;$0$}
\lbl[c]{66.66,66.66;$0$}
\lbl[c]{92   ,66.66;$0$}
\lbl[c]{117.4,66.66;$1$}
\lbl[c]{142.8,66.66	;$1$}
\end{lpic}}
{\caption{Function $\HN$.}\label{fig:hn}}
\end{floatrow}
\end{figure}

\subsection{Two-party example} \label{ssec:versatility}

Consider the function $\VER\colon \Z_4\times\Z_4\to\{0,1\}$ defined by
\begin{equation}
\VER(x,y) = 1\enspace \iff\enspace x+y \in \{2,3\},\qquad\text{for all } x,y\in\Z_4,
\end{equation}
where the arithmetic is that of $\Z_4$; see \autoref{fig:ver}.
\begin{lemma} \label{lem:g}
$\VER$ is versatile.
\end{lemma}
\begin{proof}
The reduction from $\AND$ is simply given by $\AND(x,y) = \VER(x,y)$. Moreover, $\VER$ is flippable because $\neg \VER(x,y)= \VER(x+2,y)$. To see that $\VER$ is random-self-reducible, start with $(x,y)$ and compute as follows. First, choose $(\bm{x},\bm{y})$ uniformly at random from the set $\{(x,y),(1-x,-y)\}$ so that $\bm{x}+\bm{y}$ is uniformly distributed either in the set $\{0,1\}$ if $(x,y)$ was a $0$-input, or in the set $\{2,3\}$ if $(x,y)$ was a $1$-input. Finally, choose a random $\bm{a}\in\Z_4$ and output $(\bm{x}+\bm{a},\bm{y}-\bm{a})$.
\end{proof}

It is not hard to show that $\VER$ is in fact a minimum-size example of a versatile function: if $g\colon[a]\times[b]\to\{0,1\}$ is versatile then $a,b\geq 4$. Indeed, $\VER$ is the smallest two-party function for which our proof of \autoref{thm:two-party} applies. By comparison, the original proof of \autoref{thm:two-party}~\cite{huynh12virtue} uses a certain subfunction $\HN\leq\ThreeIND$ whose communication matrix is illustrated in \autoref{fig:hn}. Thus, somewhat interestingly, our proof yields a result that is incomparable to~\cite{huynh12virtue} since we have neither $\VER\leq\HN$ nor $\HN\leq\VER$.

Coincidentally, $\VER$ makes an appearance in Sherstov's pattern matrix method~\cite[\S12]{sherstov11pattern}, too. There, the focus is on exploiting the \emph{matrix-analytic} properties of the communication matrix of $\VER$. By contrast, in this work, we celebrate its \emph{self-reducibility} properties.

\subsection{Multi-party examples} \label{ssec:multi-party-examples}

In the multi-party setting we restrict our attention to $k$-party reductions $f_1\leq f_2$ for $k$-party functions $f_i\colon\mathcal{X}_i^k\to\{0,1\}$ that are determined by one-to-one mappings $\pi_1,\ldots,\pi_k$ satisfying
\[
f_1(x_1,\ldots,x_k)=f_2(\pi_1(x_1),\ldots,\pi_k(x_k))\qquad\text{for all}\quad(x_1,\ldots,x_k)\in \mathcal{X}_1^k.
\]
This way any player that sees an input $x_i$ can evaluate $\pi_i(x_i)$ without communication. As before, a randomised reduction can also depend on public coins.

\emph{Versatile} $k$-party functions $g_k\colon\mathcal{X}^k\to\{0,1\}$ are defined analogously to the two-party case: we require that the $k$-party $k$-bit $\AND_k$ function reduces to $g_k$, and that $g_k$ is both flippable and random-self-reducible---all under $k$-party reductions.

It is known that every $k$-party function is a subfunction of some, perhaps exponentially large random-self-reducible function~\cite{feige94minimal}. However, in the following, we are interested in finding examples of \emph{small} versatile $k$-party functions in order to optimise our constructions. We proceed to give two examples of well-studied $k$-party functions and prove them versatile.

\paragraph{First example: Quadratic character.}

Denote by $\chi\colon\Z_p^\times\to\{0,1\}$ the indicator function for quadratic residuosity modulo $p$, i.e., $\chi(x) = 1$ iff $x$ is a square in $\Z_p$. The pseudo-random qualities of $\chi$ have often made it an object of study in communication complexity~\cite{babai92multiparty,babai03communication,ada12nof}. Moreover, the self-reducibility properties of $\chi$ are famously useful in cryptography, starting with~\cite{goldwasser84probabilistic}.

For our purposes we let $p$ to be an $O(k)$-bit prime. Following~\cite[\S2.5]{babai92multiparty} the $k$-party quadratic character function $\QCS_k\colon\Z_p^k\to\{0,1\}$ is defined as 
\begin{equation}
\textstyle
\QCS_k(x_1,\ldots,x_k):=\chi\big(\sum_i x_i\big).
\end{equation}
We leave $\QCS_k(x_1,\ldots,x_k)$ undefined for inputs with $\sum_i x_i = 0$, i.e., we consider $\QCS_k$ to be a promise problem. Our three items of versatility fall out of the well-known properties of~$\chi$.
\begin{lemma}\label{qcs-versatile}
$\QCS_k$ is versatile.
\end{lemma}
\begin{proof}
\emph{Reduction from $\AND_k$:}
We need the following nonelementary fact (see, e.g., Lemma~6.13 in~\cite{babai03communication} or the recent work~\cite{wright13quadratic}): if $p$ is a large enough $O(k)$-bit prime then there are $k+1$ consecutive integers $\{a,a+1,\ldots,a+k\}\subseteq \Z^\times_k$ realising the pattern
\[
\chi(a) = \chi(a+1) = \cdots = \chi(a+k-1) = 0\qquad\text{and}\qquad \chi(a+k)=1.
\]
This immediately facilitates the reduction: an input $(y_1,\ldots,y_k)$ of $\AND_k$ is mapped to an input $(a+y_1,y_2,\ldots,y_k)$ of $\QCS_k$.
\emph{Flippability:}
Map $x_i\mapsto s\cdot x_i$ for all $i$, where $s\neq 0$ is a fixed quadratic nonresidue.
\emph{Random-self-reducibility:}
Choose a random quadratic residue $\bm{r}\in\Z_p$ and numbers $\bm{a}_1,\ldots,\bm{a}_k\in\Z_p$ satisfying $\bm{a}_1+\cdots+\bm{a}_k = 0$. The random self-reduction maps $x_i\mapsto \bm{r}\cdot x_i + \bm{a}_i$ for all $i$.
\end{proof}

\paragraph{Second example: Pointer jumping.}

Next, we observe that certain variants of the $k$-party pointer jumping function are versatile. To explain this idea, we begin by describing a simple construction where each of the $k$ inputs requires $\Theta(k\log k)$ bits to represent. After this we improve on the construction by using known results on branching programs; we note that similar ideas have been used in the context of secure multi-party computations~\cite{cramer03efficient}.

Define the $k$-party \emph{pointer jumping} function $\Jump_k\colon \mathcal{X}^k\to\{0,1\}$ as follows. The inputs are permutations $x_i\colon [2k]\to[2k]$, $i\in[k]$, and the function value is given by
\begin{equation} \label{eq:jump-k}
\Jump_k(x_1,\ldots,x_k) = 0 \quad\iff\quad (x_k\circ x_{k-1}\circ \cdots \circ x_1)(1) \in [k].
\end{equation}
A useful way to view the input $(x_1,\ldots,x_k)$ is as a layered digraph: there are $k+1$ layers, each containing $2k$ nodes; the input $x_i$ defines a perfect matching between layers $i$ and $i+1$; and the nodes on the last layer are labelled in a \emph{balanced} way with $k$ zeroes and $k$ ones. The value of the function is the label of the sink that is reachable from the $1$st node of the $1$st layer.
\begin{lemma} \label{lem:jump}
$\Jump_k$ is versatile.
\end{lemma}
\begin{proof}
\emph{Reduction from $\AND_k$:}
Given an input $(y_1,\ldots,y_k)$ of $\AND_k$ we reduce it to an input $(x_1,\ldots,x_k)$ of $\Jump_k$ as follows (see \autoref{fig:pointer}). If $y_i=0$ then $x_i$ is defined to be the identity permutation on $[2k]$, otherwise $x_i$ is the cyclic permutation that maps $j\mapsto j+1$ for $j\in[2k-1]$ and $2k\mapsto 1$.
\emph{Flippability:}
Replace the input $x_k$ with $\pi\circ x_k$, where $\pi\colon[2k]\to[2k]$ is some fixed permutation that swaps the sets $[k]$ and $[k+1,2k]$, i.e., $\pi([k])=[k+1,2k]$.
\emph{Random-self-reducibility:}
The random self-reduction is best visualised as acting on the layered graph associated with an input $(x_1,\ldots,x_k)$. First, sample $k+1$ permutations $\bm{\pi}_1,\ldots,\bm{\pi}_{k+1}\colon[2k]\to[2k]$ uniformly and independently at random under the restrictions that $\bm{\pi}_1$ fixes the element $1$ and $\bm{\pi}_{k+1}$ fixes the set $[k]$. Then use $\bm{\pi}_i$ to relabel the nodes on the $i$-th layer. Formally this means that the input $x_i$ is mapped to $\bm{\pi}_{i+1}\circ x_i\circ \bm{\pi}^{-1}_i$.
\end{proof}

The reduction $\AND_k\leq\Jump_k$ above was implicitly using a simple read-once permutation branching program for $\AND_k$; see \autoref{fig:pointer}. We will now
optimise this construction by using more efficient branching programs.

\begin{figure}[t]
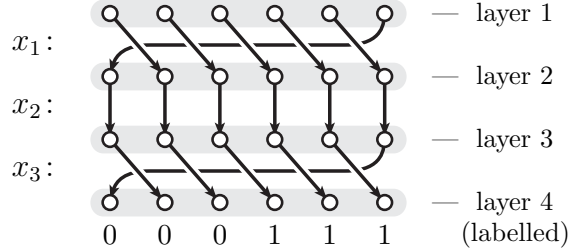

\floatbox[{\capbeside\thisfloatsetup{capbesideposition={left,top},capbesidewidth=6.5cm}}]{figure}[\FBwidth]
{\caption{Example of $\AND_3\leq\Jump_3$. The input $(x_1,x_2,x_3)$ of $\Jump_3$ is the result of applying the reduction to the input $(1,0,1)$ of $\AND_3$.}
\label{fig:pointer}}
{\begin{lpic}[l(1.4cm),r(2cm),t(-.5cm)]{pointer(.33)}
\small
\lbl[l]{140,103.7;{---~~layer $1$}}
\lbl[l]{140,78.3;{---~~layer $2$}}
\lbl[l]{140,52.8;{---~~layer $3$}}
\lbl[l]{140,27.5;{---~~layer $4$}}
\lbl[l]{151.5,14.9;{(labelled)}}
\large
\lbl[r]{-8,91;$x_1\colon$}
\lbl[r]{-8,65.5;$x_2\colon$}
\lbl[r]{-8,40.2;$x_3\colon$}
\normalsize
\lbl[c]{10.2,14.9;$0$}
\lbl[c]{32.3,14.9;$0$}
\lbl[c]{54.5,14.9;$0$}
\lbl[c]{76.6,14.9;$1$}
\lbl[c]{98.8,14.9;$1$}
\lbl[c]{120.9,14.9;$1$}
\end{lpic}}
\end{figure}

\begin{definition}[PBPs]
A \emph{permutation branching program} (PBP) of width $w$ and length $\ell$ is defined by a sequence of instructions $(i_l,\pi_l,\tau_l)$, $l\in[\ell]$, where $\pi_l,\tau_l\colon[w]\to[w]$ are permutations and each $i_l\in[n]$ indexes one of the $n$ input variables $x_1,\ldots,x_n$. Let an input $x\in\{0,1\}^n$ be given. We say that an instruction $(i,\pi,\tau)$ \emph{evaluates to $\pi$} if $x_i=0$; otherwise the instruction \emph{evaluates to~$\tau$}. The PBP \emph{evaluates} to the composition of the permutations evaluated at the instructions. Finally, if $\gamma\colon[w]\to[w]$ is a permutation, we say that the PBP \emph{$\gamma$-computes} a function $f\colon\{0,1\}^n\to\{0,1\}$ if it evaluates to the identity permutation $e\colon[w]\to[w]$ on each $0$-input in $f^{-1}(0)$ and to the permutation $\gamma\neq e$ on each $1$-input in $f^{-1}(1)$.
\end{definition}

\begin{lemma} \label{lem:pbp-to-versatile}
Suppose there exists a width-$w$ length-$\ell$ PBP that $\gamma$-computes the $\AND_k$ function. Then there exists a versatile $k$-party function on $O(\ell w \log w)$ input bits.
\end{lemma}
\begin{proof}
Fix a width-$w$ PBP $(i_l,\pi_l,\tau_l)$, $l\in[\ell]$, that $\gamma$-computes $\AND_k$. By modifying the PBP if necessary, we may assume that $w$ is even and $\gamma(1)\in[w/2+1,w]$. The versatile function corresponding to the given PBP is the pointer jumping function $\Jump_k^\ell(x_1,\ldots,x_\ell)$ defined similarly to (\ref{eq:jump-k}):
\[
\Jump_k^\ell(x_1,\ldots,x_\ell) = 0 \quad\iff\quad (x_\ell\circ x_{\ell-1}\circ \cdots \circ x_1)(1) \in [w/2].
\]
To define the input partition, let $L_i:=\{l\in[\ell] : i_l = i\}$ be the set of layers where the PBP reads the $i$-th input. We let the $i$-th player hold (on its forehead) the inputs $x_l$ for $l\in L_i$.

\emph{Reduction from $\AND_k$:}
The reduction $\AND_k\leq \Jump_k^\ell$ is naturally determined by the PBP: given an input $(y_1,\ldots,y_k)$ of $\AND_k$, we define $x_l$ to be the permutation that the instruction $(i_l,\pi_l,\tau_l)$ evaluates to under $(y_1,\ldots,y_k)$. Because of our input partition, it is possible to compute $x_l$ without communication.

\emph{Flippability and random-self-reducibility:}
Same as in the proof of \autoref{lem:jump}.
\end{proof}
Barrington's celebrated theorem~\cite{barrington89bounded} gives a PBP implementation of $\AND_k$ with parameters $w=5$ and $\ell=O(k^2)$. This corresponds to having $O(k)$ input bits per player, matching the quadratic character example above. Cleve~\cite{cleve91towards} has improved this to a tradeoff result where for any $\epsilon>0$ one can take $\ell=k^{1+\epsilon}$ provided that $w=w(\epsilon)$ is a large enough constant. Cleve's construction also has the property that every input variable of $\AND_k$ is read equally many times (i.e., the $L_i$ in the above proof have the same size). Thus, letting $w$ grow sufficiently slowly, we get a versatile $k$-party gadget on $O(\ell w\log w)= k^{1+o(1)}$ bits, which is~$k^{o(1)}$ bits per player.
\begin{corollary} \label{cor:gadget}
There are versatile $k$-party gadgets $g_k\colon\mathcal{X}^k\to\{0,1\}$ where $\log|\mathcal{X}| = k^{o(1)}$. \qed
\end{corollary}

% =========================================================================== %
\section{Communication Lower Bound} \label{sec:cc-lb}
% --------------------------------------------------------------------------- %

In this section we prove the communication lower bound for two parties (\autoref{thm:two-party}) assuming that $g$ is a versatile gadget. The generalisation to multiple parties (\autoref{thm:multi-party}) follows by the same argument---one only needs to replace $g$ with a versatile $k$-party gadget $g_k$.

Our proof builds on a result of Zhang~\cite{zhang09tightness} that lower bounds the two-party communication complexity of a composed function $f\circ g^n$ in terms of the block sensitivity of $f$. We start by outlining Zhang's approach.

\subsection{Functions: Zhang's approach} \label{ssec:zhang}

Zhang~\cite{zhang09tightness} proved the following theorem by a reduction from the \emph{unique-disjointness} function $\UDISJ_n$. Here, $\UDISJ_n = \OR_n\circ\AND^n$ is the usual set-disjointness function together with the promise that if $\UDISJ_n(a,b) = 1$, then there is a unique coordinate $i\in[n]$ such that $a_i=b_i=1$. The randomised communication complexity of $\UDISJ_n$ is well-known to be $\Theta(n)$~\cite{kalyanasundaram92probabilistic,razborov92distributional,bar-yossef04information}. Zhang's proof works for any gadget $g$ with $\AND,\OR\leq g$.%
\begin{theorem}[Zhang]
There is a two-party gadget $g\colon \mathcal{X}\times\mathcal{Y}\to\{0,1\}$ such that if $f\colon\{0,1\}^n\to Q$ is a function, then $f\circ g^n$ has communication complexity $\Omega(\bs(f))$.
\end{theorem}
The proof runs roughly as follows. Fix an input $\alpha\in\{0,1\}^n$ for $f$ that witnesses the block sensitivity $\bs(f,\alpha)=\bs(f)$. Also, let $B_1,\ldots,B_\bs\subseteq[n]$ be the sensitive blocks of $f$ at~$\alpha$. Given an input $(a,b)$ to $\UDISJ_\bs$ the goal in the reduction is for the two parties to compute, without communication, an input $(x,y)$ for $f\circ g^n$ such that
\begin{description}[noitemsep]
\item[(T1)] \emph{0-inputs:} If $\UDISJ_\bs(a,b)=0$, then $g^n(x,y) = \alpha$.
\item[(T2)] \emph{1-inputs:} If $\UDISJ_\bs(a,b)=1$ with $a_i=b_i=1$, then $g^n(x,y)=\alpha^{B_i}$.
\end{description}
Clearly, if we had a reduction $(a,b)\mapsto(x,y)$ satisfying (T1--T2), then the output of $\UDISJ_\bs(a,b)$ could be recovered from $(f\circ g^n)(x,y)$. Thus, an $\epsilon$-error protocol for $f\circ g^n$ would imply an $\epsilon$-error protocol for $\UDISJ_\bs$ with the same communication cost.

\subsection{Search problems: Our approach}

We are going to prove \autoref{thm:two-party} (restated below) in close analogy to the proof template (T1--T2) above. However, as discussed below, noncritical inputs to search problems introduce new technical difficulties.
\twoparty*

\paragraph{Setup.}
Fix any versatile gadget $g\colon \mathcal{X}\times\mathcal{Y}\to\{0,1\}$.
Let $\Pi$ be a randomised $\epsilon$-error protocol for a composed search problem $S\circ g^n$. Recall that an input $(x,y)$ for the problem $S\circ g^n$ is \emph{critical} if there is exactly one solution $q$ with $((x,y),q)\in S\circ g^n$. In particular, if $g^n(x,y)$ is critical for $S$, then $(x,y)$ is critical for $S\circ g^n$. The behaviour of the protocol $\Pi$ on a critical input $(x,y)$ is predictable: the protocol's output $\Pi(x,y)$ is the unique solution with probability at least $1-\epsilon$.

However, noncritical inputs $(x,y)$ are much trickier: not only can the distribution of the output $\Pi(x,y)$ be complex, but the distributions of $\Pi(x,y)$ and $\Pi(x',y')$ can differ even if $(x,y)$ and $(x',y')$ encode the same input $g^n(x,y)=g^n(x',y')$ of $S$. The latter difficulty is the main technical challenge, and we address it by using random-self-reducible gadgets.

\paragraph{Defining a function $\bm{f\subseteq S}$.}
We start by following very closely the initial analysis in the proof of Huynh and Nordstr{\"o}m~\cite{huynh12virtue}. First, we record for each $\alpha\in\{0,1\}^n$ the \emph{most likely feasible output} of $\Pi$ on inputs $(x,y)$ that encode $\alpha$. More formally, for each $\alpha$ we define $\mu_\alpha$ to be the uniform distribution on the set of preimages of $\alpha$, i.e.,
\begin{equation} \label{eq:mu-alpha}
\mu_\alpha\enspace\text{is uniform on}\enspace \{(x,y) : g^n(x,y) = \alpha \}.
\end{equation}
Alternatively, this can be viewed as a product distribution
\begin{equation} \label{eq:mu-z}
\mu_\alpha = \mu_{\alpha_1} \times \mu_{\alpha_2} \times \cdots \times \mu_{\alpha_n},
\end{equation}
where $\mu_z$, $z\in\{0,1\}$, is the uniform distribution on $g^{-1}(z)$.

The most likely feasible solution output by $\Pi$ on inputs $(\bm{x},\bm{y})\sim\mu_\alpha$ is now captured by a total function $f\subseteq S$ defined by
\begin{equation} \label{eq:def-f}
f(\alpha)\ :=\ \argmax_{q:(\alpha,q)\in S} \Pr_{(\bm{x},\bm{y})\sim\mu_\alpha}[\,\Pi(\bm{x},\bm{y}) = q\,].
\end{equation}
Here, ties are broken arbitrarily and the randomness is taken over both $(\bm{x},\bm{y})\sim\mu_\alpha$ and the random coins of the protocol $\Pi$. (Note that, in general, the most likely output of $\Pi(\bm{x},\bm{y})$ may not be feasible. However, above, we explicitly pick out the most likely \emph{feasible} solution. Thus, $f$ is indeed a subfunction of $S$.)

\paragraph{The sensitive critical input.}
We can now use the critical block sensitivity of $S$: there is a critical input $\alpha$ such that $\bs(f,\alpha)\geq\cbs(S)$. Let $B_1,\ldots,B_\bs\subseteq[n]$ be the sensitive blocks with $f(\alpha^{B_i})\neq f(\alpha)$.
\begin{lemma} \label{lem:distinguish}
The protocol $\Pi$ can distinguish between $\mu_\alpha$ and $\mu_{\alpha^{B_i}}$ in the sense that
\begin{alignat}{3}
(\bm{x},\bm{y}) &\sim \mu_\alpha
&\quad\implies\quad& \Pr[\,\Pi(\bm{x},\bm{y})=f(\alpha)\,] \geq 1-\epsilon, \label{eq:f-crit}\\
(\bm{x},\bm{y}) &\sim \mu_{\alpha^{B_i}}
&\quad\implies\quad& \Pr[\,\Pi(\bm{x},\bm{y})= f(\alpha)\,] \leq 1/2. \label{eq:f-noncrit}
\end{alignat}
\end{lemma}
\begin{proof}
The consequent in the first property (\ref{eq:f-crit}) is true even for each individual $(x,y)$ in the support of $\mu_\alpha$ since $\alpha$ is critical. To see that the second property (\ref{eq:f-noncrit}) is true, suppose for a contradiction that we had $\Pr[\,\Pi(\bm{x},\bm{y})= f(\alpha)\,] > 1/2$ for $(\bm{x},\bm{y})\sim\mu_{\alpha^{B_i}}$. By averaging, there is a fixed input $(x,y)$ in the support of $\mu_{\alpha^{B_i}}$ such that $\Pr[\,\Pi(x,y)= f(\alpha)\,] > 1/2$. By the correctness of $\Pi$ (i.e., $1-\epsilon > 1/2$) this implies that $f(\alpha)$ is feasible for $\alpha^{B_i}$. Thus, $f(\alpha)$ is the most likely feasible solution output by $\Pi(\bm{x},\bm{y})$, that is, $f(\alpha^{B_i})=f(\alpha)$ by the definition~\eqref{eq:def-f}. But this contradicts the fact that $f$ is sensitive to $B_i$ at $\alpha$.
\end{proof}

\paragraph{The reduction.}
\autoref{lem:distinguish} suggests a reduction strategy analogous to the template (T1--T2) of \autoref{ssec:zhang}. Given an input $(a,b)$ for $\UDISJ_\bs$ our goal is to describe a randomised reduction $(a,b)\mapsto(\bm{x},\bm{y})$ such that
\begin{description}[noitemsep]
\item[(P1)] \emph{0-inputs:} If $\UDISJ_\bs(a,b)=0$, then $(\bm{x},\bm{y})\sim \mu_\alpha$.
\item[(P2)] \emph{1-inputs:} If $\UDISJ_\bs(a,b)=1$ with $a_i=b_i=1$, then $(\bm{x},\bm{y})\sim \mu_{\alpha^{B_i}}$.
\end{description}

Suppose for a moment that we had a reduction with properties (P1--P2). Let $\Pi'$ be the protocol that on input $(a,b)$ first applies the reduction $(a,b)\mapsto(\bm{x},\bm{y})$ with properties (P1--P2), then runs $\Pi$ on $(\bm{x},\bm{y})$, and finally outputs $0$ if $\Pi(\bm{x},\bm{y})=f(\alpha)$ and $1$ otherwise. \autoref{lem:distinguish} tells us that
\begin{itemize}[label=$-$,noitemsep]
\item If $\UDISJ_\bs(a,b)=0$, then $\Pi'(a,b)=0$ with probability at least $1-\epsilon$.
\item If $\UDISJ_\bs(a,b)=1$, then $\Pi'(a,b)=1$ with probability at least $1/2$.
\end{itemize}
The error probability of $\Pi'$ can be bounded away from $1/2$ by repeating $\Pi'$ twice and outputting~$0$ iff both runs of $\Pi'$ output $0$. (Here we are assuming that $\epsilon$ is small enough, say at most $1/4$. If not, we can use some other standard success probability boosting tricks.) This gives a randomised protocol for $\UDISJ_\bs$ with the same communication cost (up to constants) as that of $\Pi$. \autoref{thm:two-party} follows.

Indeed, it remains to implement a reduction $(a,b)\mapsto(\bm{x},\bm{y})$ satisfying (P1--P2). We do it in three steps; see \autoref{fig:reduction}.

\begin{figure}[t]
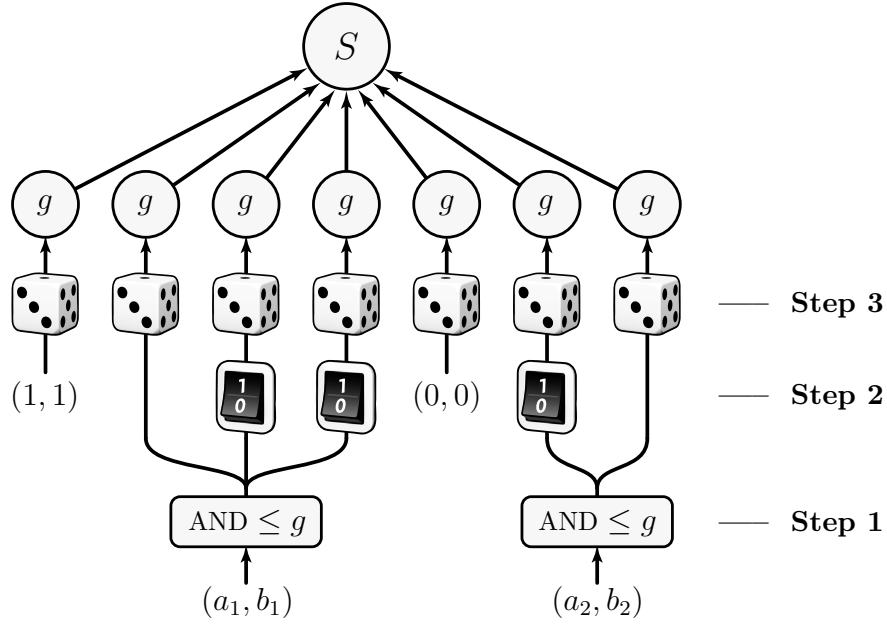
%
\ffigbox[.9\textwidth]
{\caption{The reduction $(a,b)\mapsto(\bm{x},\bm{y})$. In this example $\bs=2$ and $n=7$. The critical input is $\alpha=1011010$ and the two sensitive blocks are $B_1=\{2,3,4\}$ and $B_2=\{6,7\}$. The input pair $(a_i,b_i)$, $i=1,2$, is plugged in for the block $B_i$.}
\label{fig:reduction}}
{\begin{lpic}[r(2cm),t(3mm),b(3mm)]{reduction(.35)}
\large
\lbl[t]{91,23;$(a_1,b_1)$}
\lbl[t]{224,23;$(a_2,b_2)$}
\lbl[c]{91,47;$\AND\leq g$}
\lbl[c]{224,47;$\AND\leq g$}
\lbl[c]{14,95;$(1,1)$}
\lbl[c]{166.7,95;$(0,0)$}
\lbl[c]{14,168;$g$}
\lbl[c]{52.2,168;$g$}
\lbl[c]{90.3,168;$g$}
\lbl[c]{128.5,168;$g$}
\lbl[c]{166.7,168;$g$}
\lbl[c]{204.3,168;$g$}
\lbl[c]{243,168;$g$}
\lbl[c]{128.5,228.5;\Large$S$}
\normalsize
\lbl[l]{270,130;{\bf-----~~Step 3}}
\lbl[l]{270,95;{\bf-----~~Step 2}}
\lbl[l]{270,47;{\bf-----~~Step 1}}
\end{lpic}}
\hspace{5mm}
\end{figure}

\paragraph{Step 1.}
On input $(a,b)=(a_1\ldots a_\bs,b_1\ldots b_\bs)$ to $\UDISJ_\bs$ we first take each pair $(a_i,b_i)$ through the reduction $\AND\leq g$ to obtain instances $(a'_1,b'_1),\ldots,(a'_\bs,b'_\bs)$ of $g$. Note that
\begin{itemize}[label=$-$,noitemsep]
\item if $\UDISJ_\bs(a,b)=0$, then $g(a'_i,b'_i)=0$ for all $i$;
\item if $\UDISJ_\bs(a,b)=1$, then there is a unique $i$ with $g(a'_i,b'_i)=1$.
\end{itemize}

\paragraph{Step 2.}
Next, the instances $(a'_i,b'_i)$ are used to populate a vector $(x,y)=(x_1\ldots x_n,y_1\ldots y_n)$ carrying $n$ instances of $g$, as follows. The instance $(a'_i,b'_i)$ is plugged in for the coordinates $j\in B_i$ with the copies corresponding to $\alpha_j=1$ 	flipped. That is, we define for $j\in B_i$:
\begin{itemize}[label=$-$,noitemsep]
\item if $\alpha_j = 0$, then $(x_j,y_j) := (a'_i,b'_i)$;
\item if $\alpha_j = 1$, then $(x_j,y_j) := (\pi_A(a'_i),\pi_B(b'_i))$, where $(\pi_A,\pi_B)$ is the reduction $\neg g\leq g$.
\end{itemize}
For $j\notin \cup_i B_i$ we simply fix an arbitrary $(x_j,y_j)\in g^{-1}(\alpha_j)$. We now have that
\begin{itemize}[label=$-$,noitemsep]
\item if $\UDISJ_\bs(a,b)=0$, then $g^n(x,y) = \alpha$;
\item if $\UDISJ_\bs(a,b)=1$ with $a_i=b_i=1$, then $g^n(x,y)= \alpha^{B_i}$.
\end{itemize}

\paragraph{Step 3.}
Finally, we apply a random-self-reduction independently for each component $(x_i,y_i)$ of $(x,y)$: this maps a $z$-input $(x_i,y_i)$ to a uniformly random $z$-input $(\bm{x}_i,\bm{y}_i)\sim \mu_z$. The result is a random vector $(\bm{x},\bm{y})$ that has a distribution of the form (\ref{eq:mu-z}) and matches our requirements (P1--P2), as desired.

\bigskip\noindent
This concludes the proof of \autoref{thm:two-party}. The proof of the multi-party version (\autoref{thm:multi-party}) is exactly the same, except with $g$ and $\UDISJ_\bs$ replaced by a versatile $g_k$ and $\UDISJ_{k,\bs}$. Here, in particular, $\UDISJ_{k,n}$ is the usual $k$-party disjointness function $\DISJ_{k,n}=\OR_n\circ\AND^n_k$ together with the promise that at most one of the $\AND_k$'s evaluates to $1$.

\pagebreak[2]
% =========================================================================== %
\section{Critical Block Sensitivity Lower Bounds} \label{sec:cbs}
% --------------------------------------------------------------------------- %

In this section we prove our new critical block sensitivity bounds, Theorems \ref{thm:tseitin} and \ref{thm:pebbling}.

\subsection{Tseitin sensitivity} \label{sec:tseitin}

Let $G=(V,E,\ell)$ be a connected graph with an odd-weight labelling $\ell\colon V\to\{0,1\}$. Recall that in the problem $S(\Tse_G)$ the input is an assignment $\alpha\colon E\to\{0,1\}$ and the goal is to find a parity violation, that is, a node in $\Viol(\alpha) := \{v\in V : C_v(\alpha) = 0\}$.

For the readers' convenience, we recall some basic facts about $\Tse_G$. Since each edge $e\in E$ participates in two constraints, the sum $\sum_v\sum_{e:v\in e} \alpha(e)$ will be even. By contrast, the sum $\sum_v\ell(v)$ is odd. It follows that $|\Viol(\alpha)|$ must be odd, and, in particular, non-empty. Conversely, for every odd-size set $U\subseteq V$, there is an $\alpha$ with $\Viol(\alpha)=U$. To see this, start with any assignment $E\to\{0,1\}$ and let $p$ be a simple path in $G$. If we flip the truth values of the edges in $p$, we end up flipping whether or not the constraints at the endpoints of $p$ are satisfied. Depending on whether the endpoints of $p$ were satisfied to begin with, this results in one of the following scenarios: (1) we create a pair of violations; (2) we remove a pair of violations; or (3) we move a violation from one endpoint of $p$ to the other. It is not hard to see that by using (1)--(3) repeatedly, we can design an assignment $\alpha$ such that $\Viol(\alpha)=U$.

We are now ready to prove \autoref{thm:tseitin}.

\tseitin*
\begin{proof}
Let $G=(V,E,\ell)$ be $(\kappa+1)$-routable. Fix a set $T\subseteq V$ of size $|T|=2\kappa+1$ such that whenever $M$ is a set of $\kappa$ disjoint pairs of nodes from $T$, there are $\kappa$ edge-disjoint paths connecting each pair in $M$. We denote by $\Paths(M)$ some canonical set of such paths.

Consider the following bipartite auxiliary graph on \emph{left} and \emph{right} vertices:
\begin{itemize}[label=$-$]
\item {\bf Left vertices} are pairs $(\alpha,M)$, where $\alpha\colon E\to\{0,1\}$ has a \emph{unique} violation that is in $T$ (i.e., $|\Viol(\alpha)| = 1$ and $\Viol(\alpha)\subseteq T$), and $M$ is a partition of the set $T \smallsetminus \Viol(\alpha)$ into $\kappa$ pairs of nodes.
\item {\bf Right vertices} are pairs $(\alpha',M')$, where $\alpha'\colon E\to\{0,1\}$ has \emph{three} violations that are all in $T$ (i.e., $|\Viol(\alpha')| = 3$ and $\Viol(\alpha')\subseteq T$), and $M'$ is a partition of the set $T \smallsetminus \Viol(\alpha')$ into $\kappa-1$ pairs of nodes.
\item {\bf Edges} are defined as follows. A left vertex $(\alpha,M)$ is connected to a right vertex $(\alpha',M')$ if $M' \subseteq M$ and $\alpha'$ is obtained from $\alpha$ by flipping the values along the path in $\Paths(M)$ that connects the pair $\Viol(\alpha')\smallsetminus\Viol(\alpha)$.
\end{itemize}
The key fact, which is easy to verify, is that the auxiliary graph is \emph{biregular}: its left-degree is $\kappa$ and its right-degree is $3$.

To prove the block sensitivity bound, let $f$ be a function solving $S(\Tse_G)$. We say that an edge from $(\alpha,M)$ to $(\alpha',M')$ in the auxiliary graph is \emph{sensitive} if $f(\alpha) \neq f(\alpha')$. Clearly, for each right vertex exactly two (out of three) of its incident edges are sensitive. Thus, by averaging, we may find a left vertex $(\alpha,M)$ such that at least a fraction $2/3$ of its incident edges are sensitive. But this means that $\alpha$ is a critical input with block sensitivity at least~$2\kappa/3$; the blocks are given by a subset of $\Paths(M)$.
\end{proof}

\subsection{Pebbling sensitivity}

\begin{figure}[t]
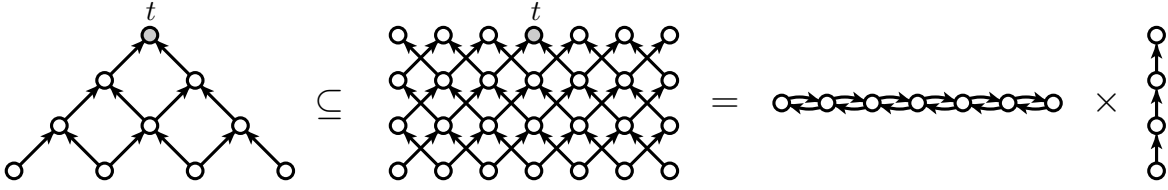
%
{\caption{Pyramid graph viewed as a subgraph of a tensor product of paths.}
\label{fig:pyramid}}
{\begin{lpic}{pyramid(.474)}
\normalsize
\lbl[c]{42,50;$t$}
\lbl[c]{150,50;$t$}
\Large
\lbl[c]{92,24;$\subseteq$}
\lbl[c]{203,23.55;$=$}
\lbl[c]{310,24.3;$\times$}
\end{lpic}}
\end{figure}

\pebbling*

\paragraph{Overview.}
Our proof of \autoref{thm:pebbling} generalises the original proof from~\cite{huynh12virtue} that held for pyramid graphs. The key idea is natural: In a pyramid graph, each horizontal layer can be interpreted as a path---this is made precise by viewing the pyramid graph as a subgraph of a tensor product of paths as in \autoref{fig:pyramid}. The analysis in the original proof suffered from the fact that random walks do not mix well on paths. So, we replace the paths by graphs with better mixing properties! (Perhaps surprisingly, we do not need to rely on expanders here.)

\paragraph{Definition of $\bm{G}$.}
Let $H$ be the $3$-dimensional grid graph on $m=r^3$ nodes where $r$ is odd. For convenience, we think of $H$ as a directed Cayley graph on $\Z_r^3$ generated by the $6$ elements
\[
\mathcal{B}=\{\pm (1,0,0),\pm (0,1,0),\pm (0,0,1)\}.
\]
That is, there is an edge $(v,u)\in E(H)$ iff $u = v+b$ for some $b\in\mathcal{B}$.
The key property of $H$ (which is not satisfied by $d$-dimensional grid graphs for $d<3$) is the following.
\begin{lemma}[Partial cover time] \label{lem:pct}
Starting from any node of $H$ the expected number of steps it takes for a random walk to visit at least half of the nodes of $H$ is $\pct(H)=O(m)$.
\end{lemma}
\begin{proof}
This follows from Lemma~2.8 in~\cite{lovasz93random} and the fact that the maximum hitting time of $H$ is $O(m)$ (e.g.,~\cite{chandra96electrical}).
\end{proof}

Let $\ell:=2\cdot\pct(H)+1=\Theta(m)$ so that by Markov's inequality a random walk of length $\ell-1$ in $H$ will cover a at least a fraction $1/2$ of $H$ with probability at least $1/2$. Let $P$ be the directed path on $[\ell]$ with edges $(i,i+1)$, $i\in[\ell-1]$. We construct the tensor product graph
\[
G\coloneqq H\times P
\]
that is defined by $V(G)=\Z_r^3\times[\ell]$ and there is a directed edge from $(v,i)$ to $(u,j)$ iff $j=i+1$ and $u = v + b$ for some $b\in\mathcal{B}$.

The $n=m\ell$ nodes of $G$ are naturally partitioned into~$\ell$ \emph{layers} (or \emph{steps}). In order to turn $G$ into a pebbling formula, we need to fix some \emph{sink} node $t$ in the $\ell$-th layer and delete all nodes from which $t$ is not reachable. We do not let this clean-up operation affect our notations, though. For example, we continue to think of the resulting graph as $G=H\times P$. The nodes $\Z_r^3\times\{1\}$ of indegree $0$ will be the \emph{sources}.

Note that each source--sink path $p$ in $G$ contains exactly one node from each layer. We view the projection of $p$ onto $H$ as a walk of length $\ell-1$ in $H$; we can describe the walk uniquely by a sequence of $\ell-1$ generators from $\mathcal{B}$. We denote by $\pi(p)\subseteq V(H)$ the set of nodes visited by the projected walk.

We can now study the search problem $S(\Peb_G)$ associated with the pebbling formula~$\Peb_G$.

\paragraph{Pebbling number.}
The pebbling strategy for $G$ that uses $O(\sqrt{n})=O(m)$ pebbles proceeds as follows. We first pebble the $1$st layer (the sources), then the $2$nd layer, then remove pebbles from the $1$st layer, then pebble the $3$rd layer, then remove pebbles from the $2$nd layer, etc.

The matching lower bound follows from the fact that $G$ contains a pyramid graph on $\Omega(n)$ nodes as a subgraph, and the pebbling number of pyramid graphs is $\Theta(\sqrt{n})$~\cite{cook74observation}.

\paragraph{Decision tree complexity.}
The deterministic decision tree that uses $O(\sqrt{n})=O(m)$ queries proceeds as follows. We start our search for a violated clause at the sink $t$. If the sink variable is false, we query its children to find a child $v$ whose associated variable is false. The search continues at $v$ in the same manner. In at most $\ell-1=O(m)$ steps we find a false node $v$ whose children are all true (perhaps $v$ is a source node).

The matching lower bound follows from the critical block sensitivity lower bound proved below, and the fact that critical block sensitivity is a lower bound on the decision tree complexity.

\paragraph{Critical block sensitivity.}
It remains to prove that $\cbs(S(\Peb_G))=\Omega(m)$. The following proof is a straightforward generalisation of the original proof from (the full version of)~\cite{huynh12virtue}.

All \emph{paths} that we consider in the following are source--sink paths in $G$. We associate with each path $p$ a critical input $\alpha_p\colon V(G)\to\{0,1\}$ that assigns to each node on $p$ the value $0$ and elsewhere the value $1$. This creates a unique clause violation at the source where $p$ starts.

If $p$ and $q$ are two paths, we say that $p$ and $q$ are \emph{paired at $i\geq 2$} if the following hold.
\begin{itemize}[label=$-$,noitemsep]
\item \emph{Agreement:} $p$ and $q$ do not meet before layer $i$, but they agree on all layers $i,\ldots,\ell$.
\item \emph{Mirroring:} if the first $i-1$ steps of $p$ are described by $(b_1,b_2,\ldots,b_{i-1})\in\mathcal{B}^{i-1}$, then the first $i-1$ steps of $q$ are described by $(-b_1,-b_2,\ldots,-b_{i-1})\in\mathcal{B}^{i-1}$.
\end{itemize}
Each path can be paired with at most $\ell-1$ other paths---often, there are plenty such:
\begin{lemma} \label{lem:plenty}
Each path $p$ is paired with at least $|\pi(p)|-1$ other paths.
\end{lemma}
\begin{proof}
For each node $v\in \pi(p)$, except the starting point of $p$, we construct a pair $q$ for $p$. To this end, let $i\geq2$ be the first step at which the projection of $p$ visits $v$. Since the mirroring property uniquely determines $q$ given $p$ and $i$, we only need to show that this~$q$ satisfies the agreement property. Thus, suppose for a contradiction that $p$ and $q$ meet at some node $(u,j)$ where $j<i$. We have, in $\Z_r^3$ arithmetic,
\begin{align*}
v &= u + b_j + b_{j+1} + \cdots + b_{i-1} \qquad(\text{according to $p$}),\\
v &= u - b_j - b_{j+1} - \cdots - b_{i-1} \qquad(\text{according to $q$}).
\end{align*}
This implies $2v=2u$, but since $r$ is odd, we get $v=u$. This contradicts our choice of $i$.
\end{proof}

If $p$ and $q$ are paired, we can consider the assignment $\alpha_{p\cup q}$ that is the node-wise logical $\AND$ of the assignments $\alpha_p$ and $\alpha_q$. In $\alpha_{p\cup q}$ we have \emph{two} clause violations associated with the two starting points of the paths.

To prove the critical block sensitivity bound $\Omega(m)$, let $f$ be a function solving $S(\Peb_G)$. Consider the following auxiliary graph.
\begin{itemize}[label=$-$,noitemsep]
\item The {\bf vertices} are the source--sink paths.
\item There is a {\bf directed edge} from $p$ to $q$ iff $p$ and $q$ are paired and $f(\alpha_{p\cup q})$ is the starting point of $q$. Thus, each two paired paths are connected by an edge one way or the other.
\end{itemize}
Recall that if we start a random walk of length $\ell-1$ on $H$ at any fixed node, the walk covers a fraction $\geq 1/2$ of $H$ with probability $\geq 1/2$. If we view a source-sink path~$p$ in $G$ in the \emph{reverse order} (starting at the sink and going towards the source), this translates into saying that $|\pi(p)|\geq m/2$ for a fraction $\geq 1/2$ of all paths $p$. Applying \autoref{lem:plenty} for such paths we conclude that the auxiliary graph has average outdegree at least $d=m/8-1$. By averaging, we can now find a path~$p$ with out-neighbours $q_1,\ldots,q_d$. Define $q'_i \coloneqq q_i \smallsetminus p$. Clearly the critical assignment $\alpha_p$ is sensitive to each $q'_i$. To see that the $q'_i$ are pairwise disjoint, we note that they take steps in the same direction in $\mathcal{B}$ at each layer (i.e., opposite to that of $p$), and the $q_i$ meet $p$ for the first time at distinct layers. This concludes the proof of \autoref{thm:pebbling}.

% =========================================================================== %
\section{Monotone CSP-SAT} \label{sec:monotone-proofs}
% --------------------------------------------------------------------------- %

In this section we introduce a monotone variant of the CSP satisfiability problem and show how lifted search problems $S(F)\circ g^n$ reduce to its monotone Karchmer--Wigderson game. In particular, in \autoref{cor:monotone-dept-tseitin} we can take the explicit function to be a CSP satisfiability function. We also note that our function has been further studied by Oliveira~\cite[Chapter~3]{oliveira15unconditional}.

\paragraph{Definition of monotone CSP-SAT.}
The function is defined relative to some finite alphabet~$\Sigma$ and a fixed constraint topology given by a bipartite graph $G$ with left vertices $V$ (\emph{variable nodes}) and right vertices $U$ (\emph{constraint nodes}). We think of each $v\in V$ as a variable taking on values from $\Sigma$; an edge $(v,u)\in E(G)$ indicates that variable $v$ is involved in constraint node $u$. Let $d$ be the maximum degree of a node in $U$. We define $\SAT=\SAT_{G,\Sigma}\colon\{0,1\}^N\to\{0,1\}$ on $N\leq |U|\cdot |\Sigma|^d$ bits as follows. An input $\alpha\in\{0,1\}^N$ describes a CSP instance by specifying, for each constraint node $u\in U$, its \emph{truth table}: a list of at most $|\Sigma|^d$ bits that record which assignments to the variables involved in $u$ satisfy~$u$. Then $\SAT(\alpha)\coloneqq 1$ iff the CSP instance described by $\alpha$ is satisfiable. This encoding of CSP satisfiability is indeed monotone: if we flip any 0 in a truth table of a constraint into a 1, we are only making the constraint easier to satisfy.

\subsection{Reduction to CSP-SAT}

Recall the characterisation of monotone depth due to Karchmer and Wigderson~\cite{karchmer88monotone}: if $f\colon\{0,1\}^N\to\{0,1\}$ is a monotone function, then its monotone depth complexity is equal to the (deterministic) communication complexity of the following search problem.
\begin{quote}
{\bf Monotone KW-game for $\bm{f}$:} Alice holds a $a\in f^{-1}(1)$ and Bob holds a $b\in f^{-1}(0)$. The goal is to find a coordinate $i\in[N]$ such that $a_i=1$ and $b_i=0$.
\end{quote}
The next lemma shows that for any search problem of the form $S(F)\circ g^n$ there is a some monotone CSP-SAT function whose monotone KW-game embeds $S(F)\circ g^n$. (The reduction can be seen as a generalisation of Lemma 3.5 in \cite{raz99separation}.) 

We define the \emph{constraint topology} of $F$ naturally as the bipartite graph $G$ with left vertices $\vars(F)$ and right vertices $\cons(F)$. For a constraint $C\in\cons(F)$ we use the lower case $c$ to denote the corresponding node in $G$ (forgetting that $C$ is actually a function).

\begin{lemma} \label{lem:csp-sat}
Let $g\colon\mathcal{X}\times\mathcal{Y}\to\{0,1\}$ be a two-party gadget and let $F$ be an unsatisfiable $d$-CSP on $n$ variables and $m$ constraints. Let $G$ be the constraint topology of $F$. Then the monotone depth complexity of $\SAT_{G,\mathcal{X}} \colon \{0,1\}^N\to\{0,1\}$, $N\leq m|\mathcal{X}|^d$, is lower bounded by the (deterministic) communication complexity of $S(F)\circ g^n$.
\end{lemma}
\begin{proof}
We reduce the search problem $S(F)\circ g^n$ to the monotone KW-game for $\SAT_{G,\mathcal{X}}$. To this end, let $(x,y)$ be an input to the search problem $S(F)\circ g^n$ and compute as follows.
\begin{itemize}[label=$-$]
\item Alice maps $x\in \mathcal{X}^{\vars(F)}$ to the CSP whose sole satisfying assignment is $x$. That is, the truth table for a constraint node $c$ is set to all-0 except for the entry indexed by $x\upharpoonright\vars(C)$ (restriction of $x$ to the variables in $C$).
\item Bob maps $y\in\mathcal{Y}^{\vars(F)}$ to an unsatisfiable CSP as follows. The truth table for a constraint node $c$ is such that the bit indexed by $\ell\in{\mathcal{X}}^{\vars(C)}$ is set to $1$ iff $C$ is satisfied under the partial assignment $v\mapsto g(\ell(v),y(v))$ where $v\in\vars(C)$.
\end{itemize}

Alice clearly constructs a $1$-input of $\SAT_{G,\mathcal{X}}$. To see that Bob constructs a $0$-input of~$\SAT_{G,\mathcal{X}}$, suppose for a contradiction that there is a global assignment $\ell\colon\vars(F)\to\mathcal{X}$ so that the truth table of each $c$ has a $1$ in position indexed by $\ell\upharpoonright\vars(C)$. This would mean that the truth assignment $v\mapsto g(\ell(v),y(v))$ satisfies all the constraints of $F$. But this contradicts the unsatisfiability of $F$.

Assume then that Alice and Bob run a protocol for the monotone KW-game on the CSP instances constructed above. The output of the protocol is a some entry $\ell\in\mathcal{X}^{\vars(C)}$ in the truth table of some constraint node $c$ where Alice has a $1$ and Bob has a $0$. Because Alice's CSP was constructed so that for each constraint node~$c$ exactly one entry is $1$, we must have that $\ell = x\upharpoonright\vars(C)$. On the other hand, Bob's construction ensures that $C$ is not satisfied under the assignment $v \mapsto g(\ell(v),y(v))=g(x(v),y(v))$. Thus, we have found a violated constraint $C$ for the canonical search problem for $F$.
\end{proof}

\subsection{Proof of \autoref{cor:monotone-dept-tseitin}}
Theorems \ref{thm:two-party} and \ref{thm:tseitin} yield a search problem $S(\Tse_G)\circ g^m$ of communication complexity $\Omega(n/\log n)$ where $G$ is an $n$-node $m$-edge bound-degree graph ($d=O(1)$, $m=O(n)$) and $g$ is a constant-size gadget ($|\mathcal{X}|=O(1)$). Using \autoref{lem:csp-sat} we can then construct a CSP-SAT function on $N=O(n)$ bits having monotone depth $\Omega(n/\log n)=\Omega(N/\log N)$.
This proves \autoref{cor:monotone-dept-tseitin}.

% =========================================================================== %
\section{Applications: Proof Complexity} \label{sec:proof-complexity}
% --------------------------------------------------------------------------- %

In this section we prove our new proof complexity lower bounds as stated in \autoref{sec:app-proof}. We start by reviewing some standard proof complexity terminology.

\subsection{Background}

In this work we focus on proof systems that refute unsatisfiable CNF formulas. Given a proof system, a \emph{refutation} (or a \emph{proof}) $P$ of an unsatisfiable CNF formula $F$ in the system is expressed as a sequence of {\em lines}, denoted $\Lines(P)$, each of which is either (a translation of) a clause of $F$ or follows from some previous lines via some sound {\em inference rule}. The refutation ends with some trivially false line.

For each proof $P$ we can associate a directed acyclic graph $G_P=(V,E)$ where $V=\Lines(P)$ and there is an edge $(u,v)\in E$ if $v$ is derived via some inference rule using line $u$.

\paragraph{Complexity measures.}
For the purposes of this work, we define the \emph{size} of a proof $P$ simply as the number of lines $|\Lines(P)|$. The \emph{rank} of $P$ is the length of the longest path in $G_P$. The \emph{size complexity} and \emph{rank complexity} of $F$ in a proof system are the minimum size and minimum rank, respectively, of all refutations of $F$ in that system.

We consider $G_P$ to be a tree if every internal node has fan-out one, that is, the clauses of $F$, which are not internal nodes, can be repeated.
If $G_P$ is a tree, we say that $P$ is \emph{tree-like}. The \emph{tree-like size complexity} of $F$ is the minimum size of a tree-like refutation of $F$. Note that restricting a refutation to be tree-like does not increase the rank because each line can be re-derived multiple times without affecting the rank. Tree-like size, however, can be much larger than general size.

\paragraph{Examples of proof systems.}
We mention some of the most well-studied proof systems.
In each of these systems, there is a set of derivation rules (which can be
thought of as inference schemas) of the form $F_1,F_2,\ldots, F_t\vdash F_{t+1}$
and each inference in a proof must be an instantiation of one of these rules.

A basic system is \emph{Resolution} whose lines are clauses.
Its only rule is the \emph{resolution rule}: the clause $(A \lor B)$
can be derived from $(A \lor x)$ and $(B \lor \neg x)$, where $A$ and
$B$ are arbitrary disjunctions of literals and $x$ is a variable.
A Resolution refutation of an unsatisfiable CNF formula $f$
is a sequence of clauses, ending with the empty clause,
such that each clause in the sequence is either a clause of $f$,
or follows from two previously derived clauses via the resolution rule.

Another proof system is the \emph{Cutting Planes} ($\CP$) proof system that manipulates integer linear inequalities. A $\CP$ refutation is a sequence
of inequalities, ending with $0 \geq 1$, such that all inequalities
are either translations of clauses of $F$, or follow from
two previously derived inequalities via one of the
two $\CP$ rules, addition and division with rounding. There is a natural extension of $\CP$, denoted $\CP(k)$, in
which the above $\CP$ proof rules may also be applied when the lines are
allowed to be degree $k$ multivariate polynomials.

Other important well-studied proof systems are the \LSfull
proof systems ($\LS_0$, $\LS$, $\LS_+$, and $\LS_{+,\star}$)
which are dynamic proof systems that manipulate polynomial inequalities of degree at most 2;
the Sherali--Adams and Lasserre (SOS) systems that are static proof systems
allowing polynomial inequalities of higher degree; and the 
dynamic Lasserre (dynamic SOS), and $\LS^{k}_{+,\star}$ systems, which generalize
the \LSfull systems to higher degree.
We refer the reader to~\cite{odonnell13approximability} for formal definitions and a thorough history 
for these and related proof sytems.

\paragraph{Semantic proof systems.}
Each of the above proof systems has a specific set of inference rule schemas, which allows them to have polynomial-time verifiers. In this work we consider more powerful \emph{semantic} proof systems that restrict the form of the lines and the fan-in of the inferences but dispense with the requirement of a polynomial-time verifier and allow any semantically sound inference rule with a given fan-in. The fan-in must be restricted because the semantic rules are so strong. The following system was introduced
in~\cite{beame07lower}.

\begin{definition}[Degree $k$ threshold proofs]
We denote by $\Th{k}$ the semantic proof system whose proofs have fan-in 2 and each line in a refutation of a formula $F$ is a polynomial inequality of degree at most $k$ in the variables of $F$. In particular, each clause of $F$ enters the system as translated into a linear inequality (similarly to the $\CP$ system discussed above).
\end{definition}

The following lemma follows from Caratheodory's Theorem.
\begin{lemma}\label{caratheodory-app}
$\CP$ and $\LS$ proofs can be efficiently converted into $\Th{k}$ proofs:
\begin{itemize}
\item Any $\CP$ proof of size (tree-like size) $s$ and rank $r$ can be converted
to a $\Th{1}$ proof of size (tree-like size) $O(s)$ and rank $O(r\log s)$.
\item Any $\LS_0$, $\LS$, or $\LS_+$ proof of size (tree-like size) $s$ and rank
$r$ can be converted to a $\Th{2}$ proof of size (tree-like size) $O(s)$ and
rank $O(r\log s)$.
\end{itemize}
\end{lemma}
Moreover, it is not hard to show that one can extend the above
simulations by $\Th{k}$ proofs to $\CP(k)$, $\LS^{k}_{+,\star}$, and
degree $k$ (dynamic) Lasserre proofs.

%Moreover, it is not hard to show that one can extend the above simulations
%by $\Th{k}$ proofs to $\CP{k}$ and $LS^{k}_{+,\star}$.
%
%The Sherali-Adams and Lasserre proof systems introduce new variables for 
%subsets of input variables of bounded size (which is called the rank of such
%proofs).   Monomials of degree $k$ represent the intended meaning of these
%added variables so $\Th{k}$ proofs of rank $k$
%also efficiently simulate rank $k$ Sherali-Adams proofs and rank $k/2$
%Lasserre proofs.
%(PLEASE CHECK THIS.  I AM NOT SURE EXACTLY WHAT OUR RESULTS SAY ABOUT THESE
%PROOFS.)

In this paper we consider semantic proof systems that are even more general than $\Th{k}$, namely those for which the fan-in is bounded and the truth value of each line can be computed by an efficient multi-party NOF communication protocol.

\begin{definition}[Proofs with $k$-party verifiers]
We denote by $\Tcc{k,c}$ the semantic proof system of fan-in $2$ in which
each proof line is a boolean function whose value, for every $k$-partition of the input variables, can be computed by a $c$-bit randomised $k$-party NOF protocol of error at most $1/4$.  Both $k=k(s)$ and $c=c(s)$ may be functions of $s$, the size of the input formula. In keeping with the usual notions of what constitutes efficient communication, we use $\Tcc{k}$ to denote $\Tcc{k,\polylog s}$.
\end{definition}

Note that via standard boosting, we can replace the error $1/4$ in the above definition by $\epsilon$ at the cost of increasing $c$ by an $O(\log 1/\epsilon)$ factor. Therefore, without loss of generality, in the definition of $\Tcc{k}$ we can assume that the error is at most $2^{-\polylog s}$.

%Note also that a semantic proof of rank $r$ that satisfies the same conditions as a $\Tcc{k,c}$ proof except that it has rules of fan-in at most $t\ge 2$ can be simulated by a $\Tcc{k,2ct\log t}$ proof of rank $r\log t$ by replacing each inference by a binary tree of height $\log t$ in which lines of internal nodes are conjunctions of their predecessors.

For polylogarithmic $k$, the following lemma
shows that $\Th{k}$ is a subclass of $\Tcc{k+1}$.

\begin{lemma} \label{th-to-rcc}
Every $\Th{k}$ refutation of an $n$-variable CNF formula is a $\Tcc{k+1,O(k^3\log^2 n)}$ refutation.
\end{lemma}
\begin{proof}
By the well-known result of Muroga~\cite{muroga71threshold}, linear threshold functions on
$n$ boolean variables only require coefficients of $O(n \log n)$ bits.
Since a degree $k$ threshold polynomial is a linear function on at most $n^k$
monomials, it is equivalent to a degree $k$ threshold polynomial
with coefficients of $O(kn^k \log n)$ bits.
As shown in~\cite{beame07lower},
over any input partition
there is a randomized $(k+1)$-party communication protocol of cost
$O(k\log^2 b)$ 
and error $\le 1/b^{\Omega(1)}$ 
to verify a degree $k$ polynomial inequality with $b$-bit coefficients.
\end{proof}

The following lemma, which is implicit in~\cite{beame07lower},
gives the key relationships between $\Tcc{k}$ 
and randomised communication protocols for $S(F)$.

\begin{lemma} \label{proof-to-com}
If a CNF formula $F$ has a $\Tcc{k,c}$ refutation of rank $r$ then, over any $k$-partition of the variables, there is a randomised bounded-error $k$-party NOF protocol for $S(F)$ with communication cost $O(c\cdot r\log r)$.
\end{lemma}

\subsection{Lifting CNF formulas}

In order to import our communication lower bounds to proof complexity, we need to encode composed search problems $S\circ g^n_k$ as CNF formulas. We describe a natural way of doing this in case $S=S(F)$ is the search problem associated with some CNF formula $F$.

Fix a $d$-CNF formula $F$ on $n$ variables and $m$ clauses. Also, fix a $k$-party gadget $g_k\colon\mathcal{X}^k\to\{0,1\}$ where each player holds $l\coloneqq\log|\mathcal{X}|$ bits as input. We construct a new $D$-CNF formula $F\circ g^n_k$ on $N$ variables and $M$ clauses, where
\begin{equation} \label{eq:params}
D=d\cdot kl,\qquad N=n\cdot kl,\qquad\text{and}\qquad M\leq m\cdot 2^{dkl}.
\end{equation}
\paragraph{Variables of $\bm{F\circ g_k^n}$.}
For each variable $x$ of $F$ we create a matrix of variables
\[
X=\{\, X_{ij} : i \in [k],\, j \in [l]\,\}.
\]
The idea is that truth assignments $\alpha_X \colon X \to\{0,1\}$ are in a natural one-to-one correspondence with the set $\mathcal{X}^k$, the domain of $g_k$. Namely, the value of the $j$-th bit of the $i$-th player is encoded by $X_{ij}$. We take the variable set of $F\circ g_k^n$ to be the union $X\cup Y \cup \ldots$, where $x,y,\ldots$ are the original variables of $F$.

\paragraph{Clauses of $\bm{F\circ g_k^n}$.}
Let $C$ be a clause of $F$; suppose first that $C=(x\lor \neg y)$ for simplicity. We will replace $C$ with a set of clauses $\mathcal{C}$ on the variables $X\cup Y$ such that all clauses of $\mathcal{C}$ are satisfied under an assignment $\alpha\colon X\cup Y\to\{0,1\}$ if and only if $g_k(\alpha_X) = 1$ or $g_k(\alpha_Y) = 0$; here $\alpha_X$ and $\alpha_Y$ are elements of $\mathcal{X}^k$ associated with the restrictions of $\alpha$ to $X$ and $Y$. Indeed, let $X^\alpha_{ij} = X_{ij}$ if $\alpha(X_{ij}) = 1$, and $X^\alpha_{ij} = \neg X_{ij}$ if $\alpha(X_{ij}) = 0$, and similarly for $Y^\alpha_{ij}$. Define a clause
\[
C_\alpha =
\Big(\neg \bigwedge_{i,j} X^\alpha_{ij}\Big) \vee
\Big(\neg \bigwedge_{i,j} Y^\alpha_{ij}\Big),
\]
and let $\mathcal{C}$ consist of all the clauses $C_\alpha$ where $\alpha$ is such that $g_k(\alpha_X)=0$ and $g_k(\alpha_Y)=1$.

More generally, if we had started with a clause on $d$ variables, each clause $C_\alpha$ would involve $dkl$ variables and so we would have $|\mathcal{C}| \leq 2^{dkl}$. This completes the description of $F\circ g_k^n$.

\bigskip\noindent
The formula $F\circ g_k^n$ comes with a natural partition of the variables into $k$ parts as determined by the $k$-party gadget. Thus, we can consider the canonical search problem $S(F\circ g_k^n)$.
\begin{lemma} \label{lem:lifted-search}
The two problems $S(F\circ g_k^n)$ and $S(F)\circ g_k^n$ have the same $k$-party communication complexity up to an additive $dkl$ term.
\end{lemma}
\begin{proof}
As discussed above, the \emph{inputs} to the two problems are in a natural one-to-one correspondence. How about translating \emph{solutions} between the problems? Given a violated clause $C_\alpha$ in the problem $S(F\circ g_k^n)$, it is easy to reconstruct $C$ from $C_\alpha$ without communication. Moreover, given a violated clause $C$ of $F$ in the problem $S(F)\circ g_k^n$, we can construct a violated $C_\alpha$ by first finding out what encoding $\alpha$ was used for each of the $d$ variables of $C$. This can be done by communicating $dkl$ bits (even in the number-in-hand model).
\end{proof}

\subsection{Rank lower bounds}

We are now ready to prove \autoref{thm:rank-lb}, restated here for convenience.
\rankbounds*
\begin{proof}
We start with a Tseitin formula $F$ with $n$ variables, $O(n)$ clauses, and width $O(1)$ that is associated with a $\Omega(n/\log n)$-routable bounded-degree graph. Let $k=k(n)$ be a parameter. We construct the formula $F\circ g_k^n$ where $g_k^n\colon\mathcal{X}^k\to\{0,1\}$ is the gadget of \autoref{cor:gadget}. Recall that $\log|\mathcal{X}|=k^{\epsilon}$ where $\epsilon=\epsilon(k)\to 0$ as $k\to\infty$. Using~\eqref{eq:params}, we observe
\begin{itemize}[label=$-$,noitemsep]
\item $F\circ g_k^n$ has size $s=O(n)\cdot \exp(O(k^{1+\epsilon}))$,
\item $F\circ g_k^n$ has width $O(k^{1+\epsilon})$,
\item $S(F\circ g_k^n)$ has $k$-party NOF communication complexity $\mathsf{CC}=\Omega(\sqrt{n/\log n}/2^kk)$; this follows from \autoref{lem:lifted-search}, Theorems \ref{thm:multi-party} and \ref{thm:tseitin}, and Sherstov's lower bound~\cite{sherstov13communication}. (Alternatively, the complexity is $\Omega(n/\log n)$ in case $k=2$.)
\end{itemize}
Fix $\delta>0$ and choose $k=(\log n)^{1-\delta}$. For large $n$, the above bounds translate into:
\[
s= n^{1+o(1)},\qquad
\text{width}\leq \log n,\qquad\text{and}\qquad
\mathsf{CC} \geq n^{1/2-o(1)}.
\]
Therefore, by \autoref{proof-to-com}, there are no $\Tcc{k}$ refutations of $F\circ g_k^n$ with rank at most $n^{1/2-o(1)}/\polylog n=n^{1/2-o(1)}$. The result follows by letting $\delta\to 0$ sufficiently slowly.
\end{proof}

%As a corollary to the above theorem, we obtain another proof of polynomial rank lower bounds for dynamic SOS proofs of degree $(\log s)^{1-o(1)}$.

\subsection{Length--space lower bounds} \label{sec:pebbling}

In order to study the space that is required by a refutation, we need to switch to a more appropriate \emph{space-oriented} view of proofs.
\begin{definition}[{Space-oriented proofs. E.g., \cite[\S2.2]{nordstrom13pebble}}] A refutation of a CNF formula $F$ in \emph{length} $L$ and \emph{space} $\Sp$ is a sequence of \emph{configurations} $\D_0,\ldots,\D_L$ where each $\D_i$ is a set of lines (of the underlying proof system) satisfying $|\D_i|\leq \Sp$ and such that $\D_0 = \emptyset$, $\D_L$ contains a trivially false line, and $\D_i$ is obtained from $\D_{i-1}$ via one of the following derivation steps:
\begin{itemize}[label=$-$,noitemsep]
\item {\bf Clause download:} $\D_i = \D_{i-1}\cup \{v_C\}$ where $v_C$ is a translation of some clause $C$ of $F$.
\item {\bf Inference:} $\D_i = \D_{i-1} \cup \{v\}$ where $v$ follows from some number of lines of $\D_{i-1}$ by an inference rule of the system.
\item {\bf Erasure:} $\D_i = \D_{i-1} \smallsetminus \{v\}$ for some $v\in\D_{i-1}$.
\end{itemize}
\end{definition}

Huynh and Nordstr{\"o}m \cite{huynh12virtue} proved that if $F$ has a  $\Tcc{2}$ refutation of short 
length and small space, then there is a low-cost randomised two-party protocol for~$S(F)$. 
It is straightforward to show that this result holds more generally for $\Tcc{k}$ proofs and $k$-party protocols. The high level idea is that the players can use the refutation of $F$ to do a binary search for a violated clause.

\begin{lemma}[Simulation of space-bounded proofs]\label{hn-sizespace}
Fix a CNF formula $F$ of size $s$ and some $k$-partition of its variables. If $F$ has a $\Tcc{k}$ refutation of length $L$ and space $\Sp$, then there is a $k$-party randomised bounded-error protocol for $S(F)$ of communication cost
\[
\Sp\cdot \log L\cdot \polylog s.
\]
\end{lemma}
\begin{proof}
Let $\alpha\colon \vars(F)\to\{0,1\}$ be an input to the search problem $S(F)$.
Fix a length-$L$ space-$\Sp$ refutation of $F$ with configurations $\D_0,\ldots,\D_L$.

We will describe a $k$-party protocol to find a clause of $F$ that is violated under $\alpha$. The $k$ players first consider the configuration $\D_{L/2}$ in the refutation and communicate in order to evaluate the truth value of all lines in $\D_{L/2}$ under $\alpha$. If all lines of $\D_{L/2}$ are true, they continue their search on the subderivation $\D_{L/2},\ldots,\D_L$, and otherwise the search continues on the subderivation $\D_0,\ldots,\D_{L/2}$. In this way, we do a binary search, always maintaining the invariant that the first configuration in the subderivation evaluates to true, but some line in the last configuration evaluates to false. After $\log L$ steps, the players will find an $i \in [L]$ such that all of $\D_{i-1}$ evaluates to true but some line in $\D_i$ is false under $\alpha$. By the soundness of the proof system, the false line in $\D_i$ must have been a download of a some clause of $F$ and this clause solves the search problem.

Let us analyse the communication complexity of the protocol. The cost of evaluating any particular configuration with error at most $(4\log L)^{-1} \leq (4s)^{-1}$ is $\Sp\cdot \polylog s$. Thus the overall cost is $\Sp\cdot \log L\cdot \polylog s$ and the total error is at most $1/4$.
\end{proof}

Huynh and Nordstr{\"o}m proceeded to construct formulas $\Peb_G$ of size $s$ such that they admit Resolution refutations of size $O(s)$, but for which any $\Tcc{2}$ refutation in space $\Sp$ and length $L$ must satisfy $\Sp\cdot\log L = s^{1/4-o(1)}$. Using our multi-party lower bounds, we can now generalise this tradeoff result to $\Tcc{k}$ proof systems. Namely, we prove the following result, which was stated in the introduction.
%As a corollary, we obtain length--space lower bounds (with the same parameters) for dynamic degree $k$ SOS proofs.
\tradeoffs*

\begin{proof}
The formula family, parameterised by $n\in\N$, is
\[
\Peb_G\circ g_k^n,
\]
where $G$ is the graph from \autoref{thm:pebbling} with $n$ nodes and maximum degree $d=O(1)$, and where $k=k(n)$ is a parameter, and where $g_k\colon\mathcal{X}^k\to\{0,1\}$ is again our gadget from \autoref{cor:gadget}. In particular, letting $l=\log|\mathcal{X}|$, these formulas have size
\[
s \leq \Theta(n)\cdot 2^{dkl}.
\]

\paragraph{Lower bound.}
Using $\cbs(S(\Peb_G))=\Omega(n^{1/2})$ and an argument similar to the proof of \autoref{thm:rank-lb}, we conclude that $S(\Peb_G\circ g_k^n)$ has $k$-party randomised communication complexity $\Omega(n^{1/4-o(1)})$ when we choose $k=(\log n)^{1-o(1)}$ appropriately. (Alternatively, the complexity is $\Omega(n^{1/2-o(1)})$ for $k=2$.) Recall also that with this choice of~$k$, we have $s=n^{1+o(1)}$. This proves the lower bound \eqref{eq:tradeoff} in view of \autoref{hn-sizespace}.

\paragraph{Upper bound (sketch).}
To see that the lifted formula $\Peb_G\circ g_k^n$ has a Resolution refutation of length $s^{1+o(1)}$ and space $s^{1/2+o(1)}$, we will mimic the usual length-$O(n)$ space-$O(n^{1/2})$ refutation of the original formula $\Peb_G$. This refutation follows the pebbling of $G$: whenever a node $v$, with in-neighbours $w_1,\ldots,w_d$, is pebbled, we derive the clause $(v)$ from previously derived clauses $(w_1),\ldots,(w_d)$ and the clause $(\neg w_1 \lor\cdots\lor \neg w_d \lor v)$ of $\Peb_G$.

For the lifted version $\Peb_G\circ g_k^n$ we want to do the same thing, deriving the lifted clauses associated with $(v)$ from the lifted clauses associated with $(w_1),\ldots,(w_d)$ and $(\neg w_1 \lor\cdots\lor \neg w_d \lor v)$. The number of lifted variables that underlie each pebbling step is $dkl$, and since there is always a Resolution refutation of size exponential in the number of variables, it follows that each resolution step in the original refutation of $\Peb_G$ can be simulated by $O(2^{dkl})=s^{o(1)}$ steps in the lifted proof. Thus the total length of the lifted refutation is $O(n)\cdot s^{o(1)}=s^{1+o(1)}$. Similarly, the space used is $s^{1/2+o(1)}$.
\end{proof}

\medskip
\subsection*{Acknowledgements}

We thank Yuval Filmus for pointing out the quadratic character example, and Jakob Nordstr{\"o}m and Thomas Watson for providing helpful suggestions based on an early draft of this work. Thanks to Nathan Grosshans for e-mail correspondence, which clarified our presentation of the monotone CSP-SAT function. We also thank Anil Ada, Paul Beame, Trinh Huynh, and Robert Robere for discussions, and finally the STOC reviewers for useful comments.

This research was supported in part by NSERC. The first author also acknowledges support from Alfred B.\ Lehman Graduate Scholarship.

% =========================================================================== % 
\DeclareUrlCommand{\Doi}{\urlstyle{same}}
\renewcommand{\path}[1]{\footnotesize\sf{\Doi{#1}}}

\newcommand{\etalchar}[1]{$^{#1}$}

\end{document}